\documentclass[10pt,twocolumn,twoside]{IEEEtran}

\IEEEoverridecommandlockouts                              	

\usepackage{amssymb,amsmath,amsfonts,amsthm}
\usepackage{algorithm,algorithmic}
\usepackage{cite}
\usepackage{float}
\usepackage{epsfig}
\usepackage{graphicx}
\usepackage{graphics}
\usepackage{subfigure}
\usepackage{url}
\usepackage{xspace}
\usepackage[usenames,dvipsnames]{color}
\usepackage{soul}
\usepackage{mathtools}

\floatstyle{ruled}
\newfloat{model}{H}{mod}
\floatname{model}{\footnotesize Model}
\newfloat{notatio}{H}{not}
\floatname{notatio}{\footnotesize Notation}

\newenvironment{list4}{
	\begin{list}{$\bullet$}{%
			\setlength{\itemsep}{0.05cm}
			\setlength{\labelsep}{0.2cm}
			\setlength{\labelwidth}{0.3cm}
			\setlength{\parsep}{0in} 
			\setlength{\parskip}{0in}
			\setlength{\topsep}{0in} 
			\setlength{\partopsep}{0in}
			\setlength{\leftmargin}{0.16in}}}
	{\end{list}}

\usepackage{url,changebar,bm,xspace,dsfont}
\let\mathbb=\mathds 
\def\diag{\mathop{\mathrm{diag}}}  

\newtheorem{theorem}{Theorem}

\newtheorem{defn}{Definition}
\newtheorem{prop}{Proposition}
\newtheorem{assum}{Assumption}

\newtheorem{remark}{Remark}

\ifCLASSINFOpdf
 \else
\fi

\usepackage{amsmath,amsfonts,amssymb,amscd}
\usepackage{capt-of}
\usepackage{color}
\usepackage{cite}
\usepackage{verbatim}
\usepackage{hyperref}

\hypersetup{
    colorlinks = true,
    linkcolor = [rgb]{0,0,1},
    anchorcolor = [rgb]{0,0,1},
    citecolor = [rgb]{0.9,0.5,0},
    filecolor = [rgb]{0,0,1},
    urlcolor = [rgb]{0.9,0.5,0},
    bookmarksopen = true,
    bookmarksnumbered = true,
    breaklinks = true,
    linktocpage,
    colorlinks = true,
    linkcolor = [rgb]{0.2,0.6,0.2},
    urlcolor  = [rgb]{0,0,1},
    citecolor = [rgb]{0.9,0.5,0},
    anchorcolor = [rgb]{0.2,0.6,0.2},
}

\newtheorem{lemma}{\bfseries Lemma}

\newtheorem{corollary}{\bfseries Corollary}

\newcommand{\mbb}{\mathbb}
\newcommand{\bzero}{{\mathbf{0}}}
\newcommand{\bone}{{\mathbf{1}}}
\newcommand{\scrN}{\mathcal{N}}

\begin{document}

\title{\LARGE \bf Network-Based Epidemic Control Through Optimal Travel and Quarantine Management~\thanks{Research partially supported by the NSF under grants DEB-2433726, ECCS-2317079, CCF-2200052, and IIS-1914792, by the ONR under grant N00014-19-1-2571, by the DOE under grant DE-AC02-05CH11231, by the NIH under grant UL54 TR004130, and by Boston University.}}



\author{Mahtab~Talaei, Apostolos~I.~Rikos, Alex~Olshevsky, Laura~F.~White, and Ioannis~Ch.~Paschalidis
\thanks{M. Talaei, A.I. Rikos, A. Olshevsky, and I.Ch. Paschalidis are with the Department of Electrical and Computer Engineering, and the Division of Systems Engineering, Boston University, Boston, MA, USA.  E-mails: {\tt \{mtalaei, arikos, alexols, yannisp\}@bu.edu}. I.Ch. Paschalidis is also affiliated with the Department of Biomedical Engineering and the Faculty for Computing \& Data Sciences, Boston University, Boston, MA, USA.}
\thanks{L.F. White is with the School of Public Health, Boston University, Boston, MA, USA, e-mail: {\tt lfwhite@bu.edu}.}
\thanks{The results in this paper were not presented at any conference.}
}

\maketitle
\thispagestyle{empty}
\pagestyle{empty}

\begin{abstract} 
Motivated by the swift global transmission of infectious diseases, we present a comprehensive framework for network-based epidemic control. 
Our aim is to curb epidemics using two different approaches. 
In the first approach, we introduce an optimization strategy that optimally reduces travel rates. 
We analyze the convergence of this strategy and show that it hinges on the network structure to minimize infection spread. 
In the second approach, we expand the classic SIR model by incorporating and optimizing quarantined states to strategically contain the epidemic. 
We show that this problem reduces 
to the problem of matrix balancing. 
We establish a link between optimization constraints and the epidemic's reproduction number, highlighting the relationship between network structure and disease dynamics. We demonstrate that applying augmented primal-dual gradient dynamics to the optimal quarantine problem ensures exponential convergence to the KKT point.
We conclude by validating our approaches using simulation studies that leverage public data from counties in the state of Massachusetts.
\end{abstract}


%
%
%
%

\section{Introduction}\label{sec:intro}

As the recent COVID-19 pandemic has demonstrated, effectively controlling the spread of infectious diseases is of paramount importance \cite{center2021covid, linde2009does}. 
When epidemics are viewed through the lens of interconnected networks (e.g., as cities, counties, or other societal structures), a profound challenge is to preserve the continuity of societies and minimize socio-economic disruptions. 
To address this challenge, control systems methods can play an important role \cite{nowzari2016analysis, pare2020modeling}, particularly by providing tools for designing and implementing effective control policies that can minimize the impact of epidemics on public health, economies, and societies in general. 

Network-based epidemic spread models employ a compartmental approach. 
Individuals are assigned to different compartments (e.g., susceptible, infected, recovered) based on their disease status \cite{nowzari2016analysis, pare2020modeling}.  
To control disease spreading, most works in the literature employ two types of methods: {\em optimal control} and {\em spectral optimization}.  
In the optimal control literature, the main focus is to optimize various model parameters~\cite{lee2010optimal, hayhoe2021multitask, khanafer2014optimal, 2020:Fangzhou_Buss_1115_1126}. 
These parameters are optimized with the objective of minimizing a cost function that accounts for undesirable outcomes (e.g., number of infected). 
The spectral optimization literature typically considers resource allocation problems~\cite{hota2021closed, mai2018distributed, smith2022convex, preciado2014optimal}. 
The aim is to contain the spread of the epidemic by minimizing the maximum eigenvalue of the matrix that describes the model's behavior. 
However, given the complex nature of epidemic models, designing and implementing policies to help mitigate the effects of an epidemic is computationally expensive and/or NP hard (e.g., from the perspective of spectral optimization, the optimal node and link removal problems are NP-complete and NP-hard, respectively \cite{van2011decreasing}).
Recent works focused on designing computationally efficient epidemic control strategies~\cite{2021:Ma_Olshevsky_Arxiv, 2022:Ron_Bullo_2731_2736}. 
Motivated by these works, we aim to provide computationally efficient frameworks for epidemic control that optimize allocation of available resources and minimize the impact on society. 

\noindent
\textbf{Main Contributions.}
We propose two novel strategies for network-based epidemic control. 
The first leverages optimal control principles. 
Specifically, our model's input parameters are optimized to minimize a cost function integrated along the model trajectory. 
The second strategy draws inspiration from spectral optimization techniques.
It focuses on efficient resource allocation  to minimize the spectral abscissa of the model's matrix, thereby achieving specific performance guarantees. 
Our contributions are as follows: 
\begin{list4}
    \item \textit{Optimizing Travel Rates:} We present a novel approach to suppress the spread of an epidemic by optimizing the travel rates between network nodes (see Section~\ref{sec:travel_rates}). 
    We analyze the convergence of our approach, and show that not only it effectively curbs the transmission of infectious diseases, but also it preserves societal and economic activities. 
    Specifically, its operation relies on minimizing the dominant eigenvalue of the infection spread matrix. 
    This underscores the centrality of the network analysis in epidemic control. 
    \item \textit{Optimizing Quarantine Rates:} We extend the SIR model to the SIQR model by incorporating a quarantine strategy for both asymptomatic and symptomatic individuals (see Section~\ref{sec:quarantine_rates}). For the SIQR model, we demonstrate that the optimization of quarantine costs can be formulated as a matrix balancing problem, which is solvable in polynomial time. Furthermore, we elucidate the relationship between the constraints of this problem and the basic reproduction number of the epidemic. As an extension, we show that applying augmented primal-dual dynamics to this problem achieves semi-global exponential convergence to the KKT point of the  quarantine optimization problem.
    \item \textit{Model Validation:} We validate the effectiveness of our approaches using real-world network data (cf. Sec.~\ref{sec:simulation_res}). 
    Namely, we calculate optimal travel and quarantine rates designed for the counties of Massachusetts. 
    This showcases the practical utility of our network-driven strategies. 
\end{list4}



\noindent
\textbf{Paper Organization.} 
The rest of this article is organized as follows. 
In Sec.~\ref{sec:notation} we review necessary notation and background, and in Sec.~\ref{sec:probform} we outline our problem formulation. 
Sec.~\ref{sec:travel_rates} details our approach for epidemic control by optimizing travel rates. 
In Sec.~\ref{sec:quarantine_rates} we introduce our extension of the SIR model to consider quarantined individuals, and we demonstrate how epidemic control is achievable by optimizing quarantine rates. 
Finally, in Sec.~\ref{sec:simulation_res} we present numerical results obtained from simulations on networks using publicly available data from Massachusetts counties.

\section{Notation and Preliminaries}\label{sec:notation}

The sets of real, integer, complex, and natural numbers are denoted by $\mathbb{R}$, $\mathbb{Z}$, $\mathbb{C}$, and $\mathbb{N}$, respectively. $\mathbb{R}_{\geq 0}$ ($\mathbb{R}_{>0}$) refers to nonnegative (positive) real numbers, and $\mathbb{R}_{\leq 0}$ ($\mathbb{R}_{<0}$) refers to nonpositive (negative) real numbers. Vectors and matrices are denoted by lowercase and capital letters, respectively. For a matrix $A\in \mathbb{R}^{n\times n}$, the entry at row $i$ and column $j$ is denoted by $a_{ij}$. The column vector concatenating all columns of $A$ into a single column is $a = \text{vec}(A)$. Vectors are assumed to be column vectors unless stated otherwise. $\mathbf{1}$ denotes the all-ones column vector, $I$ the identity matrix, and $\mathbf{0}$ the all-zero vector or matrix. Inequalities between vectors are elementwise. A matrix is non-negative (non-positive, negative, positive) if all its elements are non-negative (non-positive, negative, positive). A matrix is essentially non-negative if all off-diagonal elements are non-negative. The dominant eigenvalue and eigenvector of a matrix $A$ are $\lambda_{\max}(A)$ and $u_{\max}(A)$. For a vector $a$, $\|a\|$ is its Euclidean norm, and $\diag(a)$ is a diagonal matrix with $a$ as the diagonal elements. The element-wise multiplication of two vectors $a$ and $b$ is $a \circ b$. The transpose, spectral radius, inverse, and Moore-Penrose inverse of a matrix $A$ are denoted by $A^T$, $\rho(A)$, $A^{-1}$, and $A^{+}$. The complex conjugate of $A$ is $\Bar{A}$, and its conjugate transpose is $A^*$. The spectral norm of a square matrix $A$ is $\Vert A \Vert_2 = \lambda_{\max}(A^*A)^{\frac{1}{2}}$. The derivative of a variable $s$ with respect to time $t$ is $\dot{s}$. The gradient of a function $f$ is $\nabla f$. A function is $L$-smooth if it is continuously differentiable with a Lipschitz continuous gradient having Lipschitz constant $L$. A neighborhood $\scrN(Z_0)$ of a matrix $Z_0 \in \mathbb{C}^{n \times n}$ is a set of matrices close to $Z_0$ in some metric. For $x \in \mathbb{R}$, $[x]_+ := \max\{x,0\}$.

The network is represented by a {\em directed graph (digraph)}, denoted as $\mathcal{G} = (\mathcal{V}, \mathcal{E})$. In $\mathcal{G}$, the set of nodes is expressed as $\mathcal{V} = \{ v_1, \ldots, v_n \}$, and the set of edges is defined as $\mathcal{E} \subseteq \{ \mathcal{V} \times \mathcal{V} \cup {(v_i, v_i) \mid v_i \in \mathcal{V}} \}$, where each node has a virtual self-edge. The number of nodes and edges are denoted as $|\mathcal{V}| = n$ and $|\mathcal{E}| = m$, respectively. A {\em directed path} of length $t$ exists from node $v_i$ to node $v_l$ if there is a sequence of nodes $v_i \equiv l_0, l_1, \dots, l_t \equiv v_l$ such that $(l_{\tau}, l_{\tau+1}) \in \mathcal{E}$ for $\tau = 0, 1, \dots, t-1$. A digraph is called ``strongly connected'' if there is a directed path from every node $v_i$ to every other node $v_l$ for all pairs of nodes $v_i, v_l \in \mathcal{V}$. A square non-negative matrix $A\in \mathbb{R}_{\geq 0}^{n\times n}$ corresponds to a digraph $G(A)$ with $n$ nodes $\{ v_1, \ldots, v_n \}$ and a set edges $\mathcal{E}=\{ (v_i, v_j) \mid a_{ij}>0\}$.

\begin{defn}[Strongly Connected Matrix]
    A square non-negative matrix $A$ is said to be strongly connected if its corresponding digraph $G(A)$ is strongly connected.
\end{defn}

\begin{defn}[Primitive Matrix]
    A primitive matrix is a square non-negative matrix some power of which is positive.
\end{defn}

\begin{defn}[Matrix Stability and Function Convergence]
A matrix $A$ is {\em continuous time stable} if all of its eigenvalues possess real parts that are less than or equal to zero. 
A matrix $A$ is {\em discrete time stable} if all of its eigenvalues have magnitudes that are less than or equal to one. 
A function $y(t)$ exhibits decay at a rate of $\alpha$ starting at $t_0$ if the following condition holds: $y(t) \leq y(t_0) e^{- \alpha t}$ for all $t \geq t_0$.
\end{defn}

\subsection{Network spread model of COVID-19}\label{covid_network}
To model the spread of COVID-19, we use a nuanced extension of the basic Susceptible-Infected-Recovered (SIR) compartmental model. We separate the infected individuals into symptomatic and asymptomatic compartments (see \cite{2020:Giordano_Colaneri_855_860, birge2022controlling, 2021:Ma_Olshevsky_Arxiv} and references therein). Dividing the infected population into these subgroups allows us to better capture their epidemiological dynamics, such as differing infectiousness and detection rates between asymptomatic and symptomatic cases. The proposed model can be described as:
\begin{subequations}\label{original_COVID19_model}
\begin{align}
\dot{s}_i &= -s_i \sum_{j=1}^n a_{ij} (\beta^{\mathrm{a}} x_j^{\mathrm{a}} + \beta^{\mathrm{s}} x_j^{\mathrm{s}}), \label{susceptible_equation} \\ 
\dot{x}_i^{\mathrm{a}} &= s_i \sum_{j=1}^n a_{ij} (\beta^{\mathrm{a}} x_j^{\mathrm{a}} + \beta^{\mathrm{s}} x_j^{\mathrm{s}}) - (\epsilon +r^{\mathrm{a}}) x_i^{\mathrm{a}}, \label{x_asymp_equation} \\
\dot{x}_i^{\mathrm{s}} &= \epsilon x_i^{\mathrm{a}} - r^{\mathrm{s}} x_i^{\mathrm{s}}, \label{x_symp_equation} \\
\dot{h}_i&= r ^{\mathrm{a}} x_i ^ {\mathrm{a}} + r ^{\mathrm{s}} x_i ^ {\mathrm{s}} ,
\end{align}
\end{subequations}
where 
$n$ is the number of locations in our network, 
$s=(s_1,\ldots,s_n)\in \mbb{R}^n$ is a vector with elements $s_i$ equal to the proportion of the population at location $i$ being susceptible, 
$x^{\mathrm{a}}$ and $x^{\mathrm{s}}$ are vectors with elements equal to the proportion of the population at each location being asymptomatic and symptomatic,  respectively, and $h$ is the vector with elements $h_i$ denoting the portion of recovered population at node $i$.
$a_{ij}$ is the rate at which infections at location $j$ affect those at location $i$ (to be computed later), 
$\beta^{\mathrm{a}}$ and $\beta^{\mathrm{s}}$ are the disease transmission rates of asymptomatic and symptomatic individuals, respectively, 
$r^{\mathrm{a}}$ and $r^{\mathrm{s}}$ are the recovery rates of asymptomatic and symptomatic, respectively, and
$\epsilon$ is  the rate at which asymptomatic individuals develop symptoms. 

Following \cite{2021:Ma_Olshevsky_Arxiv}, different parameters are utilized for symptomatic and asymptomatic individuals due to the findings of \cite{2021:Kissler_Grad_medRxiv}.
Specifically, in \cite{2021:Kissler_Grad_medRxiv} it is shown that asymptomatic individuals exhibit a more rapid decline in viral load (compared to symptomatic). 
Consequently, not only do they experience quicker recovery, but they are also likely to be less contagious.

\begin{remark}\label{COVID_19_explanation}
    Note that we consider the model in \eqref{original_COVID19_model} due to its versatile framework. 
    While \eqref{original_COVID19_model} can be reduced to the SIR or SEIR model (by setting $\beta^{\mathrm{s}} = r^{\mathrm{s}} = \epsilon = 0$, or $\beta^{\mathrm{a}} = r^{\mathrm{a}} = 0$, respectively), its main advantage lies in its extensibility. 
    As detailed in the following sections, this extensibility is one of the main research directions of this paper.
\end{remark}

\subsection{Construction of Infection Flow and Travel Rates Matrices}\label{Travel_Matrix}

\noindent
\textbf{Infection Flow Matrix.} 
To construct the infection flow matrix $A$, we adopt the approach of \cite{birge2022controlling}, which is 
 particularly suitable for modeling the impact of lockdown measures for controlling COVID-19. 
At each location $i$, we denote the fixed population size by $N_i$. 
Additionally, we assume that people travel from location $i$ to location $j$ at a rate denoted by $\tau_{ij}$ (determining $\tau_{ij}$ is addressed later in the paper). 
Thus, the travel rate matrix is an $n \times n$ matrix defined as $T = [\tau_{ij}]_{i,j=n}^{n}$.   
The focus on travel rates is essential because they play a pivotal role in determining the evolution and spread of an epidemic. 
For instance, \cite{2006:Viboud_Grenfell_447_451} revealed a strong correlation between the regional progression of influenza and people's movements.

To construct $A$, let us track how the susceptible population moves to other compartments over time.  
The flow of the susceptible population from location $i$ to location $l$ is $s_i \tau_{il}$. 
The rate of infection at location $l$ is proportional to the fraction of total (asymptomatic and symptomatic) infected people traveling to location $l$ ($\sum_{j=1}^n N_j \tau_{jl} x_j^a$ and $\sum_{j=1}^n N_j \tau_{jl} x_j^s$, respectively) over the total population traveling to location $l$ ($\sum_{k=1}^n N_k \tau_{kl})$. People at location $i$ may get infected as they travel to location $l$, which leads to the following rate of change to the susceptible individuals at location $i$: 
\begin{equation}\label{susceptible_equation_analytic}
\dot{s}_i = - \sum_{l=1}^n s_i  \tau_{il} 
\left( \frac{\sum_{j=1}^n N_j \tau_{jl} x_j^a}{\sum_{k=1}^n N_k\tau_{kl}} \beta^a 
+ \frac{\sum_{j=1}^n N_j \tau_{jl} x_j^s}{\sum_{k=1}^n N_k \tau_{kl}} \beta^s \right). 
\end{equation}
Thus, comparing \eqref{susceptible_equation_analytic} with \eqref{susceptible_equation}, we can define 
\begin{equation}\label{travel_matrix_definition}
a_{ij} = \sum_{l=1}^n \tau_{il} \tau_{jl} \frac{N_j}{\sum_{k=1}^n N_k \tau_{kl}}.  
\end{equation} 


\noindent
\textbf{Travel Rates.} 
To construct the travel rate matrix $T$, we rely on the Human Mobility Flow dataset \cite{kang2020multiscale}.
This dataset captures daily visitor flows between {\em Census Block Groups (CBGs)} during the COVID-19 pandemic. 
It provides insights into how people move between locations, measured as the daily number of CBG-to-CBG flows for trips lasting more than $1$ minute. 
It also generates the estimates for the whole population flow that reflects the general human mobility patterns at the CBG level. 
In the simulations, we focus on county-level trips within Massachusetts. 
We aggregate the population flow counts for all nodes, including flows within each county and to other counties. 
To generate $T$, we assume people in node $i$ spend an $t_i$ ratio of their day outside of their residences and therefore,
\begin{equation*}
    \tau_{ij} =  t_i P_f(i,j)/ \sum_k {P_f(i,k)},
\end{equation*}
where $P_f$ is a population flow matrix, with $P_f(i,j)$ denoting the number of daily trips (lasting in location $j$ more than $1$ minute) from location $i$ to location $j$. 

\subsection{Stability and Splitting of Matrices}\label{matrix_stabl_splir}

We now provide the following lemmas about the stability of matrices that is important for our subsequent development. 



\begin{lemma}[A Perron–Frobenius version from \cite{2021:Ma_Olshevsky_Arxiv}]\label{perron_frob_3}
    Suppose $A$ is a matrix that is strongly connected, with nonnegative off-diagonal elements. 
    In this case, there exists a real eigenvalue of $A$ that is greater in magnitude than the real part of any other eigenvalue of $A$. 
    This dominant eigenvalue $\lambda_{\max}(A)$ is unique (simple) and its corresponding eigenvector $u_{\max}(A)$ is real and positive. 
\end{lemma} 

\begin{lemma}[\hspace{-.12cm} \cite{2021:Ma_Olshevsky_Arxiv}]\label{splitting_lemma} 
The following statements characterize continuous and discrete-time matrix stability. 
\begin{itemize}
    \item (\textit{Continuous-time stability}).
    A strongly connected matrix $P$ with non-negative off-diagonal elements is continuous-time stable if and only if there exists a vector $d > \bzero$ such that $Pd \leq \bzero$. 
    \item (\textit{Discrete-time stability}). 
    A non-negative strongly connected matrix $B$ is discrete-time stable if and only if there exists a vector $d > \bzero$ such that $Bd \leq d$.
    \item (\textit{Connection between Continuous-time and Discrete-time stability}). 
    Suppose $P = L - D$ where $L$ is a non-negative matrix while $D$ is a matrix with non-positive off-diagonal elements whose inverse is elementwise non-negative. 
    Suppose further that both $P$ and $D^{-1}L$ are strongly connected. Then $P$ is continuous-time stable if and only if $B = D^{-1}L$ is discrete-time stable.
\end{itemize}
\end{lemma}

Lemma~\ref{perron_frob_3} presents a version of the Perron-Frobenius theorem for strongly connected matrices with non-negative off-diagonal elements. 
It establishes the existence of a dominant real eigenvalue, that is unique and simple, with a positive eigenvector.
Lemma~\ref{splitting_lemma} characterizes matrix stability for both a continuous-time and a discrete-time setting, 
establishing relationships between matrix properties and stability conditions.

\subsection{Balancing of Nonnegative Matrices}\label{matrix_bal}

Matrix balancing is a mathematical problem that involves adjusting the elements of a given matrix to meet certain criteria or constraints. 
In this context, a square matrix $A^{n \times n} \in \mathbb{R}^{n \times n}_{\geq 0}$ is considered balanced if $A \mathbf{1} = A^T \mathbf{1}$. 
This means that for every $i = 1,...,n$, the sum of the entries in row $i$ of $A$ is equal to the sum of the entries in column $i$. 
The problem of balancing matrix $A$ is equivalent to finding a non-negative diagonal matrix $D$ such that $D A D^{-1}$ is balanced. 

Early works \cite{1967:Sinkhorn_Knopp_343_348} consider matrix balancing as the condition for a sequence of matrices to convergence to a doubly stochastic limit. 
More recent works \cite{2014:Rikos_Hadjicostis_190_201} aimed to balance a given matrix in a distributed way by considering that each row(/column) is the set of incoming(/outgoing) links of the corresponding node in a given network (see \cite{2016:Idel} for a survey on matrix balancing). 
Here, we rely on the approach presented in \cite{2017:Cohen_Vladu_902_913} in which the authors showed that matrix balancing can be achieved in linear time. 
Following \cite{2017:Cohen_Vladu_902_913}, we will summarize the complexity by stating that it is proportional to the number of nonzero entries of the matrix (as discussed in \cite[Sec.~$2.2$]{2021:Ma_Olshevsky_Arxiv}).

\section{Problem Formulation}\label{sec:probform}

Our goal is to optimally modify the elements of the COVID-19 model (cf.Sec.~\ref{original_COVID19_model}) in a non-uniform way to mitigate the spread of the disease while minimizing economic costs. 
We focus on two problems, denoted by \textbf{P1} and \textbf{P2}: 
\begin{itemize}
    \item[{\bf P1}] 
In the first problem, we aim to optimally modify the travel rate matrix $T$ (e.g., by placing limited travel restrictions) in order to curb the spread of infections (see Section~\ref{sec:travel_rates}). 
\item[\textbf{P2}] 
In the second problem, our objective is to extend the COVID-19 model (cf. \eqref{original_COVID19_model}) to encompass quarantined individuals. 
Subsequently, we aim to optimally select location-dependent quarantine rates in order to curb the spread of infections in the network while simultaneously minimizing the associated economic cost (see Section~\ref{sec:quarantine_rates}). 
\end{itemize}

\section{Optimal Travel Rates for Epidemic Control}\label{sec:travel_rates}



\subsection{Problem Structure}

The network spread model of COVID-19 in \eqref{original_COVID19_model} can be written in matrix form as 
\begin{equation}\label{matrix_infections}
\begin{pmatrix}
    \dot{s} \\
    \dot{x^{\mathrm{a}}} \\
    \dot{x^{\mathrm{s}}} \\
    \dot{h} \\ 
\end{pmatrix}
=
\begin{pmatrix}
    \bzero & -\beta^{\mathrm{a}} \diag{(s)} A & -\beta^{\mathrm{s}} \diag{(s)} A & \bzero \\
    \bzero & C & \beta^{\mathrm{s}} \diag{(s)} A & \bzero\\
    \bzero & \epsilon I & -r^{\mathrm{s}} I & \bzero\\
    \bzero & r^{\mathrm{a}} I & r^{\mathrm{s}} I & \bzero \\
\end{pmatrix}
\begin{pmatrix}
    s \\
    x^{\mathrm{a}} \\
    x^{\mathrm{s}} \\
    h \\
\end{pmatrix} , 
\end{equation}
where $C = \beta^{\mathrm{a}} \diag{(s)} A - (\epsilon + r^{\mathrm{a}})I$. 
Let us consider matrix $A$ as a function of $\tau = \text{vec}(T)$ (where $T$ is the travel matrix defined in Section~\ref{Travel_Matrix}). 
For \eqref{matrix_infections} we have that \cite[Assumption~$1$]{smith2022convex} and \cite[Assumption~$2$]{smith2022convex} hold. 
For this reason, in \eqref{matrix_infections} we can decouple the dynamics of $\dot{x}$ from $\dot{s}$ and $\dot{h}$, and focus on them. 
Therefore, for our analysis we focus on optimizing the travel rates in matrix $M(t_0,\tau)$ defined as 
\begin{equation}\label{inf_matrix}
M(t_0,\tau) = 
    \begin{pmatrix}
    C' & \beta^{\mathrm{s}} \diag{(s(t_0))} A(\tau) \\
    \epsilon I & -r^{\mathrm{s}} I\\
    \end{pmatrix},
\end{equation}
where $C' = \beta^{\mathrm{a}} \diag{(s(t_0))} A(\tau) - (\epsilon + r^{\mathrm{a}}) I$.  
For a fixed $t_0$ we define $f(\tau) = \lambda_{\max} \left( M(t_0, \tau) \right)$, and we aim to solve the following minimization problem: 
\begin{equation}\label{problem_formulation_P1}
    \begin{split}
        &{\min_{\tau}} \quad f(\tau) \\
        &\text{s.t.} \quad \| \tau - \tau_{0}\|_1 \leq b,\\
        & \,\,\qquad \tau \geq \bzero\, 
    \end{split}
\end{equation}
where $\tau_0$ is the vector of the initial travel rates and $b$ is a budget on the amount  of travel rate change (measured using an $\ell_1$ norm) from their initial values. 
Note that $b$ is important for preventing drastic changes over the travel rates which would result in considerable social disruption. 
This constraint is essential for addressing practical considerations related to the feasibility of changing the travel rates. 

\begin{remark} 
    In \eqref{problem_formulation_P1}, we focus on minimizing the dominant (or maximum) eigenvalue of the matrix $M(t_0,\tau)$ by optimizing the travel rates $\tau$. 
    The dominant eigenvalue plays a crucial role in characterizing the rate of infection spread within the COVID-19 model described in \eqref{matrix_infections}. 
    By minimizing the dominant eigenvalue, we effectively target the control of infection spread dynamics. 
    Additionally, by minimizing the dominant eigenvalue we can achieve a specific decay rate of infected cases. 
    For example, when the maximum eigenvalue is less than a specified threshold $-\alpha$ at time $t_0$, the infected cases will decay at a rate of at least $e^{-\alpha t}$ for times $t \geq t_0$. 
\end{remark}

\subsection{Proposed Solution}
We apply {\em Projected Gradient Descent (PGD)} by projecting onto the constraint set, specifically the non-negative orthant of the \( l_1 \) ball. To ensure convergence of the PGD, we employ a suitable stepsize choice using the backtracking line search technique. This adaptive strategy iteratively adjusts the stepsize based on sufficient decrease conditions, enhancing the robustness and convergence of the optimization process (see \cite{luenberger2003linear}). Considering these characteristics, our proposed algorithm is as follows:
\begin{equation}\label{PGD_iteration}
    \begin{split}
        &\textit{Initialization:} \\
        &\text{set} \ \gamma = 1 \ \text{and choose} \ \beta \in (0, 1), \\
        &\textit{Iteration:} \\
        &\text{while} \ f (\tau_{k} - \gamma \nabla f (\tau_{k})) > f (\tau_{k}) - \frac{\gamma}{2} \| \nabla f (\tau_{k}) \|^2, \\ 
        &\ \ \ \text{set} \ \gamma \leftarrow \beta \gamma, \\ 
        &y_{k+1} = \tau_{k} - \gamma \nabla {f (\tau_{k}}), \\
        &\tau_{k+1} = \underset{\tau \in \Omega}{\arg\min} {\| \tau - y_{k+1} \|}_2 ^2, 
    \end{split}
\end{equation}
 where $\gamma$ is the step size adjusted via backtracking line search, $\beta$ is the scaling parameter of $\gamma$,
$\tau_{k}$ is in the vector form (i.e., $\tau_k = \text{vec}(T_k)$) of the travel rates matrix $T_k$ at time step $k$, and $\Omega$ is the projection set.

Focusing on \eqref{PGD_iteration}, the update involves a gradient descent step followed by a projection onto the feasible set \( \Omega \). Each step presents unique challenges. Computing the gradient of the maximum eigenvalue of a matrix is a complex task that requires specialized techniques, which we will detail in our subsequent analysis. For the projection step, it is essential to consider appropriate constraints and transformations to ensure the optimal travel rates are practical and realistic.

\noindent
\textbf{Gradient Descent Step.} 
Computing the gradient of a matrix maximum eigenvalue may present challenges 
\cite{chen2004smooth}. 
To address them, we leverage \cite[Thm.~$2$]{1985:Magnus_191}, which we state for completeness. 

\begin{theorem}[\hspace{-.11cm} \cite{1985:Magnus_191}]\label{deriv_max_eigenvalue_calculation}
Let $\lambda_0$ be a simple eigenvalue of a matrix $Z_0 \in \mathbb{C}^{n \times n}$ and
let $u_0$ be an associated eigenvector, so that $Z_0 u_0 = \lambda_0 u_0$. 
Then a (complex) function $\lambda$ and a (complex) vector function $u$ are defined for all $Z$ in some neighborhood $\scrN(Z_0) \in \mathbb{C}^{n \times n}$ of $Z_0$, such that
$$
\lambda(Z_0) = \lambda_0, \ \ u(Z_0) = u_0 \ , 
$$ 
and 
$$
Zu = \lambda u, \ u_0 ^* u = 1, \ \text{for} \ Z \in \scrN(Z_0). 
$$
Moreover, the functions $\lambda$ and $u$ are $\infty$ times differentiable on $\scrN(Z_0)$ and the differentials at $Z_0$ are
\begin{equation}
    d\lambda = v_0 ^ * (dZ) u_0 /v_0^* u_0 \ ,
\end{equation}
and
\begin{equation}\label{du}
    du = (\lambda_0 I - Z_0) ^ {+} \left(I - \frac{u_0 v_0^{*}}{v_0^* u_0} \right) (dZ) u_0 \ , 
\end{equation}
where $v_0$ is the eigenvector associated with the eigenvalue $\Bar{\lambda}_0$ of $Z_0 ^*$, so that $Z_0 ^* v_0 = \Bar{\lambda}_0 v_0$, $d\lambda$ is the differential of the (complex) function $\lambda$, and $dZ$ is the differential of $Z$.
\end{theorem} 
Let $Z_0 = M(t_0,\tau_{k})$ in Theorem~\ref{deriv_max_eigenvalue_calculation}, with $\lambda_{\max}$ and $u_{\max}$ denoting the maximum eigenvalue and corresponding eigenvector of $M(t_0,\tau_{k})$, so $M(t_0,\tau_{k}) u_{\max} = \lambda_{\max} u_{\max}$. According to Theorem~\ref{deriv_max_eigenvalue_calculation}, in a neighborhood $\scrN(M(t_0,\tau_{k})) \in \mathbb{C}^{n \times n}$, there exists a function $f$ such that $f(M(t_0,\tau_{k})) = \lambda_{\max}$ and
\begin{equation}\label{gradient_express_M}
\frac{\partial f(M(t_0,\tau_{k}))}{\partial \tau_{ij}} = \frac{v^*_{\max} \left(\frac{\partial M(t_0,\tau_{k})}{\partial \tau_{ij}}\right) u_{\max}}{v^*_{\max} u_{\max}},
\end{equation}
where $v_{\max}$ is the eigenvector associated with $\Bar{\lambda}_{\max}$ of $M(t_0,\tau_k)^*$, satisfying $M(t_0,\tau_k)^* v_{\max} = \Bar{\lambda}_{\max} v_{\max}$. Since $M(t_0,\tau_{k})$ is real, ${M(t_0,\tau_{k})}^* = {M(t_0,\tau_{k})}^T$ and ${M(t_0,\tau_{k})} ^T v_{\max} = \lambda_{\max} v_{\max}$. By the Perron–Frobenius theorem (Lemma~\ref{perron_frob_3}), the dominant eigenvector is real and positive, thus $v^*_{\max} = v^T_{\max}$. Note that $v^T_{\max}$ and $u_{\max}$ are the left and right eigenvectors corresponding to $\lambda_{\max}$ of $M(t_0, \tau_k)$. Consequently, \eqref{gradient_express_M} becomes
\begin{equation}\label{df}
\frac{\partial f(M(t_0,\tau_{k}))}{\partial \tau_{ij}} = \frac{v^T_{\max} \left(\frac{\partial M(t_0,\tau_{k})}{\partial \tau_{ij}}\right) u_{\max}}{v^T_{\max} u_{\max}},
\end{equation}
providing the components of the gradient vector $\nabla f(M(t_0,\tau_k))$. Using \eqref{df}, we perform the gradient descent step specified in \eqref{PGD_iteration}.

\noindent
\textbf{Projection Step.} 
We now focus on the projection step of \eqref{PGD_iteration}. This step involves minimizing \( \| \tau - y_{k+1} \|_2^2 \) over \(\tau\) subject to the constraint set of~\eqref{problem_formulation_P1}, resulting in linear constraints. By converting this problem into a convex {\em Quadratic Programming (QP)} problem \cite{nocedal2006quadratic}, the optimal \(\tau\) can be efficiently found at each step.

\noindent
\textbf{Convergence Analysis.} 
For analyzing the convergence of the proposed algorithm to a stationary point and the operation of  \eqref{PGD_iteration}, we present the following lemmas which are necessary for the subsequent development.

Note that the spectral norm of the Hessian matrix $\nabla^2 f(\tau)$ is $\vert \lambda_{\max}(\nabla^2 f(\tau)) \vert$ because the Hessian matrix of $f(\tau)$ is symmetric and real.




\begin{lemma}\label{cont_eigval_M}
    The eigenvalues of the matrix $M(t_0,\tau)$ (cf. \eqref{inf_matrix}) are continuous functions of $\tau$. 
\end{lemma}

\begin{proof}
    The proof is a direct result of the continuity of the roots of the characteristic polynomial of $M(t_0,\tau)$ (see  \cite{zedek1965continuity}). 
\end{proof}

\begin{lemma}\label{L_smooth_function}
    The function $f(\tau) = \lambda_{\max} \left( M(t_0, \tau) \right)$ is $L$-smooth. 
\end{lemma}

\begin{proof}
    Before proving the lemma, note that the set consisting of all $M(t_0, \tau)$ matrices over all the possible travel rate matrices $T$ (defined in Section~\ref{Travel_Matrix}) is bounded and compact.

    Let us consider the function $f(\tau)$ representing the dominant eigenvalue of $M(t_0,\tau)$. 
    From Thm.~\ref{deriv_max_eigenvalue_calculation}, we have that $f(\tau)$ is $\infty$ times differentiable.
    Therefore, every element of the Hessian of $f(\tau)$ is continuous and differentiable. 
    From the continuity of the Hessian and continuity of the maximum eigenvalue (see Thm.~\ref{deriv_max_eigenvalue_calculation}) we have that the largest element of the Hessian achieves its maximum value on the compact set consisting of all $M(t_0,\tau)$ over all the possible travel rate matrices $T$. As a result, the spectral norm of the Hessian of $f(\tau)$ ($\Vert \nabla^2 f(x)\Vert _2$) is bounded and from~\cite{nesterov2003introductory} the maximum eigenvalue function $f(\tau)$ has a Lipschitz continuous gradient.
\end{proof}

After showing that the dominant eigenvalue function $f(\tau)$ is $L$-smooth, we are now ready to present the following convergence result for the PGD operation with backtracking line search described in \eqref{PGD_iteration}.

\begin{theorem}\label{PGD_convergence_rate}
    Let us consider the function $f(\tau_k) = \lambda_{\max} \left( M(t_0, \tau_k) \right)$.
    Executing the operation described in \eqref{PGD_iteration}, we have 
\begin{equation}
\lim_{k \rightarrow \infty} \nabla(f(\tau_k)) = 0 . 
\end{equation}
\end{theorem}

\begin{proof}
    Let $\{\tau_k\}$ be the sequence of points generated by the operation described in \eqref{PGD_iteration}. 
    According to the convergence of the backtracking line search for a differentiable and $L$-smooth function (see \cite{luenberger2003linear}), it follows that $\lim_{k \rightarrow \infty} \nabla(f(\tau_k)) = 0$.
\end{proof}

\section{COVID-19 Model with Optimal Quarantine Rates}\label{sec:quarantine_rates}

In this section, our goal is to enhance the COVID-19 model in \eqref{original_COVID19_model} by incorporating the concept of quarantine for infected individuals. 
Assigning an economic cost to each quarantined individual, we will focus on optimization strategies to reduce the quarantined population effectively, with the dual objective of containing the infection spread within our network. 

\subsection{Integrating Quarantining into a COVID-19 Model}\label{SIQR_equations_subsec}

In order to extend the COVID-19 model, we implement the following modifications:
\begin{enumerate}
    \item Asymptomatic infected individuals $x_i^{\mathrm{a}}$ at each node $i$ $(a)$ exhibit automatic recovery at a rate of $r^{\mathrm{a}}$, or $(b)$ develop symptoms at a rate of $\epsilon$, or $(c)$ are quarantined with rate $q_i^{\mathrm{a}}$, with $k_i$ denoting the proportion of the quarantined population at node $i$.
    \item Symptomatic cases $x_i^{\mathrm{s}}$ $(a)$ experience automatic recovery at a rate of $r^{\mathrm{s}}$, or $(b)$ transition to a quarantined state at a rate of $q_i^{\mathrm{s}}$. 
    \item Quarantined individuals $k_i$ recover at a rate of $r^q$.
\end{enumerate}

Considering the above modifications, the extended COVID-19 model, called the Susceptible-Infected-Quarantined-Recovered (SIQR) model, is described as follows:
\begin{subequations}
\begin{align}
\dot{s}_i &= -s_i \sum_{j=1}^n a_{ij} (\beta^{\mathrm{a}} x_j^{\mathrm{a}} + \beta^{\mathrm{s}} x_j^{\mathrm{s}}), \label{susceptible_equation2} \\ 
\dot{x}_i^{\mathrm{a}} &= s_i \sum_{j=1}^n a_{ij} (\beta^{\mathrm{a}} x_j^{\mathrm{a}} + \beta^{\mathrm{s}} x_j^{\mathrm{s}}) - (\epsilon +r^{\mathrm{a}} + q_i^{\mathrm{a}}) x_i^{\mathrm{a}} \ , \label{x_asymp_equation2} \\
\dot{x}_i^{\mathrm{s}} &= \epsilon x_i^{\mathrm{a}} - (r^{\mathrm{s}} + q_i^{\mathrm{s}}) x_i^{\mathrm{s}}, \label{x_symp_equation2}\\
\dot{k}_i &= q_i^{\mathrm{a}} x_i ^ {\mathrm{a}} + q_i^{\mathrm{s}} x_i ^ {\mathrm{s}} - r^q k_i, \\
\dot{h}_i&= r ^{\mathrm{a}} x_i ^ {\mathrm{a}} + r ^{\mathrm{s}} x_i ^ {\mathrm{s}} + r^q k_i, 
\end{align}
\label{SIQR_model}
\end{subequations}
The proposed SIQR model uses the same notation as the model in \eqref{original_COVID19_model} and only differs in the incorporation of quarantining. Quarantined individuals come from both asymptomatic and symptomatic cases and recover at a rate of $r^q$. 


\subsection{Problem Formulation for Minimizing Quarantine Rates}

Let us define the column vectors $q^{\mathrm{a}} = (q^{\mathrm{a}}_1, \ldots, q^{\mathrm{a}}_n)$, and $q^{\mathrm{s}} = (q^{\mathrm{s}}_1, \ldots, q^{\mathrm{s}}_n)$, and the  column vector of quarantine rates $q=(q^{\mathrm{a}}, q^{\mathrm{s}})$. 
We decouple the dynamics of $\dot{x}$ from the other terms in \eqref{SIQR_model} for a disease-free equilibrium and focus on analyzing the $2n \times 2n$ sub-matrix $M(t, q)$ defined as: 
\begin{equation}\label{matrix_for_convergence}
M(t, q) = 
    \begin{pmatrix}
    E & \beta^{\mathrm{s}} \diag{(s(t))} A  \\
    \epsilon I & - (r^{\mathrm{s}} I+ \diag{(q^{\mathrm{s}}})) \\ 
    \end{pmatrix}, 
\end{equation}
where $E = \beta^{\mathrm{a}} \diag{(s(t))} A - (\epsilon + r^{\mathrm{a}}) I - \diag{(q^{\mathrm{a}}})$. 
Our objective is to optimally minimize quarantine rates, thereby curbing the spread of infection within our network, while simultaneously minimizing the associated economic costs. The quarantine rates $q_i^s, q_i^a$ can be assumed to be nonnegative. We note also that there is an upper bound on what these rates could plausibly be, since the {\em inverse} of these rates equals the amount of time a typical individual would spend before being quarantined under a scenario where individuals do not move across classes such as symptomatic, asymptomatic, recovered, etc. Consequently, we may assume without loss of generality a bound of $1$ on all the quarantine rates. 
To achieve this, we formulate the following optimization problem: 
\begin{equation} 
    \begin{split} 
        &{\min}_{q} \quad \sum_{i=1}^{n} \frac{z_i^a}{1-q_i^{\mathrm{a}}} + \frac{z_i^s}{1-q_i^{\mathrm{s}}} \\
        &\text{s.t.} \quad \lambda_{\max} (M(t_0, q)) \leq -\alpha, \\
        & \qquad \bzero \leq q \leq \bone, 
    \end{split}
    \label{problem_formulation2}
\end{equation}
where $z^a = (z^a_1, \ldots, z^a_n)$ and $z^s = (z^s_1, \ldots, z^s_n)$ represent the relative economic costs associated with quarantining asymptomatic and symptomatic cases at each node, respectively. 
The objective function in \eqref{problem_formulation2} is designed to minimize the overall economic cost associated with quarantine rates. 
The economic cost for each node is proportional to the effectiveness of quarantine (i.e., in \eqref{problem_formulation2}, increasing quarantine rates leads to higher costs) and reflects a nonlinear increase in costs that seeks to model the increased complexity and logistical costs associated with quarantining a large proportion of the population. 
The summation across all nodes aggregates the economic costs, and the optimization problem aims to find the quarantine rates $q$ that minimize this combined economic cost. 
The constraints ensure that the maximum eigenvalue of the associated sub-matrix $M(t_0, q)$ is bounded by $-\alpha$. 
This reflects a condition for controlling the infection spread. 
Therefore, our objective is to minimize the economic costs associated with quarantining individuals while simultaneously achieving a decay rate of $-\alpha$ for the infected cases. 


\subsection{Proposed Solution}

The choice of optimization method for solving problem \eqref{problem_formulation2} differs from our previous approach (i.e., PGD for optimizing the travel rates) due to the nature of the constraints involved.
Specifically, unlike the linear constraints in problem \eqref{problem_formulation_P1}, the optimization problem in \eqref{problem_formulation2} incorporates nonlinear constraints on quarantine rates.
Therefore, in this case, we will utilize an alternative approach.

\noindent
\textbf{Reduction of the Optimization Problem to Weight-Balancing.}
We are now ready to show that Problem \eqref{problem_formulation2} can be reduced to a matrix balancing problem. To this end, we define the following assumption, acknowledging that this assumption is typically fulfilled by the epidemic model parameters we use. Specifically, when fitting the models to data from Massachusetts using the existing model parameters in the literature, it holds. In fact, the matrix balancing algorithm provided the correct answer to problem \eqref{problem_formulation2}.

\begin{assum} \label{assump1}
Consider the matrix $B_0$ defined as:
\begin{equation*}
B_0 = 
\begin{pmatrix}
\beta^{\mathrm{a}} \operatorname{diag}(s(t_0)) A - (\epsilon + r^{\mathrm{a}} + 1)I & \beta^{\mathrm{s}} \operatorname{diag}(s(t_0)) A \\
\epsilon I & -(r^{\mathrm{s}} + 1) I
\end{pmatrix}
\end{equation*}
and let $m$ be the largest absolute value of any diagonal element of $B_0$. Then, 
defining 
$x = \min_j \left(\epsilon  \beta^{\mathrm{s}}  s_j(t_0)  A_{jj}\right)$,
we have \[
    1 + \frac{x}{m^2} \geq m \ .
\]
\end{assum}
Assumption~\ref{assump1} is needed to ensure the solution is non-negative.
Now, we present the main theorem of our study:

\begin{theorem} \label{SIQR_unconstrained}
    The minimum quarantining problem in \eqref{problem_formulation2} can be reduced to a matrix balancing problem provided that Assumption~\ref{assump1} is satisfied, the infection flow matrix $A$ (from Sec.~\ref{Travel_Matrix}) is strongly connected, $s(t_0) > 0$, and the problem is feasible. 
\end{theorem}

\begin{proof}
See Appendix~\ref{appendix:A}. 
\end{proof}




Inspired from \cite{2021:Ma_Olshevsky_Arxiv}, we reduced problem \eqref{problem_formulation2} to a matrix balancing problem with a known polynomial complexity in terms of the number of unknown variables ($2n$). 
Note that problem \eqref{problem_formulation2} is potentially of higher-order complexity due to the eigenvalue condition in the fist constraint. 
While the matrix balancing approach necessitates calculating the inverse of a $2n \times 2n$ matrix $B_0$ (which has a cubic complexity), our proposed solution maintains a polynomial complexity with respect to the number of parameters ($2n$).

\noindent
\textbf{Connection of Constraints with the Epidemic Basic Reproduction Number.} 
In epidemiological analysis, significant emphasis is placed on determining the reproduction number $R_0$ 
In our scenario $R_0$ is of particular importance since it can indicate whether a disease outbreak will escalate into an epidemic or naturally dissipate. 
Unfortunately, calculating $R_0$ for a specific disease is a non trivial task (primarily due to the absence of comprehensive databases containing essential data such as infection rates and recovery rates). 
Motivated by this, we establish a relationship between the constraint \(\lambda_{\max} (M(t_0)) \leq -\alpha\) and the epidemic's reproduction number \(R_0\). This connection is crucial, as modifying our constraints can prevent an epidemic from escalating to a pandemic or spreading extensively within a population.

\begin{theorem}\label{reprod_theor_proof}
    The set $R_0 \leq 1$ is equivalent with $\lambda_{\max} \leq 0$
\end{theorem}
\begin{proof} 
    See Appendix~\ref{appendix:B}. 
\end{proof}
Theorem~\ref{reprod_theor_proof} implies that imposing the condition \(\lambda_{\max} (M(t_0)) \leq -\alpha\) provides a set of solutions that enforces a bound on the disease reproduction number $R_0 \leq r$, for $r \in [0,1]$.

\noindent
\textbf{Convergence via Augmented Primal-Dual Gradient Dynamics for Optimal Quarantine Control.}
We now extend our previous approach to solving the optimal quarantine problem by applying the augmented primal-dual gradient dynamics (Aug-PDGD). This extension demonstrates semi-global exponential convergence to a KKT point of the problem. Below are the relevant propositions and theorems supporting this result.

\begin{prop}\label{convexity_prop}
Assuming the feasibility of \eqref{problem_formulation2}, the cost function, defined as \( f(q) = \sum_{i=1}^{n} \left( \frac{z_i^a}{1-q_i^{\mathrm{a}}} + \frac{z_i^s}{1-q_i^{\mathrm{s}}} \right) \), is continuously differentiable and strongly convex on the feasible domain \( \bzero \leq q < \bone \). Moreover, \( M(t_0, q) \) can be expressed as the sum of an essentially non-negative matrix and a diagonal matrix with the elements of the vector \(-q\) on its diagonal. Therefore, as demonstrated in \cite{cohen1981convexity}, the dominant eigenvalue of the matrix \( M(t_0, q) \) is convex. Defining \( g_1(q) = \lambda_{\max} (M(t_0, q)) + \alpha \), \( g_{(1+i)} = -q_i \), and \( g_{2n+1+i} = q_i - 1 \) for all \( i \in \{1, \dots, 2n\} \), \( g(q) \) is continuously differentiable and convex.
\end{prop}

\begin{theorem}\label{lipschitz_constraints}
Given \( g_1(q) = \lambda_{\max} (M(t_0, q)) + \alpha \), \eqref{df} and \eqref{matrix_for_convergence} imply that
\[ \nabla g_1(q) = \frac{v_{\max} \circ u_{\max}}{v_{\max}^T u_{\max}} .\]
Recall that \( v_{\max}^T \) and \( u_{\max} \) represent the left and right eigenvectors associated with \( \lambda_{\max} (M(t_0, q)) \). Hence, \( \|\nabla g_1(q)\| \leq L_{g,1} \) (see Appendix~\ref{appendix:C} for proof), and \( g_1(q) \) is also \( M_{g,1} \)-smooth (see Lemma~\ref{L_smooth_function}). Additionally, for all \( i \geq 2 \), the gradients \( \nabla g_i \) are Lipschitz continuous and have bounded norms, which is evident from their definitions.
\end{theorem}

To analyze the augmented primal-dual gradient dynamics, we introduce the augmented Lagrangian of \eqref{problem_formulation2} as~\cite{tang2020semi}
\begin{equation}
    L_{\rho}(q, \lambda) = f(q) + \Theta_{\rho} (q, \lambda), \qquad \lambda \geq \bzero,
    \label{Aug-PDGD}
\end{equation}
where
\[
\Theta_{\rho} (q, \lambda) := \sum_{i=1}^{4n+1} \frac{[\rho g_i(q) + \lambda_i]_{+}^2 - \lambda_i^{2}}{2\rho}.
\]
The augmented primal-dual gradient dynamics are given by
\begin{subequations}
\begin{align}
    &\dot{q}(t) = -\nabla f(q(t)) - \sum_{i=1}^{4n+1}[\rho g_i(q(t)) + \lambda_i(t)]_{+} \nabla g_i(q(t)) \label{q_dot}\\
    &\dot{\lambda}(t) = \sum_{i=1}^{4n+1} \frac{[\rho g_i(q) + \lambda_i(t)]_{+} - \lambda_i(t)}{\rho} e_i.
\end{align}
\label{Aug-PDGD_diff}
\end{subequations}
\begin{prop}[\hspace{-.11cm} \cite{tang2020semi}] \label{trajectory_set}
Denoting \( y(t) = (q(t), \lambda(t)) \), \( t \geq 0 \) as a differentiable trajectory satisfying~\eqref{Aug-PDGD_diff}. Then, for a KKT point of~\eqref{Aug-PDGD}, \( y^* = (q^*, \lambda^*) \), we have \( \|y(t) - y^*\| \leq \| y(0) - y^*\| \), for all \( t \geq 0 \).
\end{prop}

\begin{theorem}\label{f_lipschitz}
The gradient of the cost function \( f(q) \) is Lipschitz continuous over the region traversed by the primal gradient dynamics described by~\eqref{Aug-PDGD_diff}.
\end{theorem}

\begin{proof}
See Appendix~\ref{appendix:D}.
\end{proof}

\begin{prop}\label{regularity_prop}
Let the active set at the local optimal solution \( q^* \) be denoted by \( \mathcal{I} := \{ i : g_i(q^*) = 0 \} \). Therefore, \( q^* \) is a regular local minimum, i.e., \( \nabla g_i(q^*) \) for \( i \in \mathcal{I} \) are linearly independent.
\end{prop}
\begin{proof}
    See Appendix~\ref{appendix:E}.
\end{proof}

\begin{corollary}\label{stability}
Under Propositions~\ref{convexity_prop} and \ref{regularity_prop}, and Theorems~\ref{lipschitz_constraints} and \ref{f_lipschitz}, the application of augmented primal-dual gradient dynamics (Aug-PDGD) to \eqref{problem_formulation2} achieves semi-global exponential stability. Specifically, the KKT point \( y^* = (q^*, \lambda^*) \) is a semi-globally exponentially stable equilibrium of the Aug-PDGD. \end{corollary}

This result directly follows immediately from Theorem \ref{f_lipschitz} and Proposition \ref{regularity_prop} and the results of ~\cite{tang2020semi}, which examines the stability of Aug-PDGD for smooth convex optimization problems with general convex and nonlinear inequality constraints. Consequently, the distance to the optimal solution decays exponentially from any initial point, although the convergence rate may depend on the initial distance to a stationary point.

It should be noted that~\cite{tang2020semi} assumes the Lipschitz continuity of \(\nabla f(q)\). However, as stated in Theorem~\ref{f_lipschitz}, what is actually required is the Lipschitz continuity of \(\nabla f(q)\) over the region invariant under the primal gradient dynamics.

\section{Simulation Results}\label{sec:simulation_res}


In this section we present simulation results for our optimization strategies. 
Specifically, for both problems \textbf{P1} and \textbf{P2} we focus on the scenario of a $14$ node network that represents the counties of the state of Massachusetts. 
We aim to demonstrate the advantages of our approaches and emphasize their applicability in real world scenarios. 

\subsection{Model Parameters}\label{mod_par_plots}

\noindent
\textbf{Massachusetts 14 Node Network.}  
The network comprises 14 nodes, each representing a county in Massachusetts. The population of each county, $N_i$ for $i \in \{1, ..., 14\}$, is based on the 2020 Census data \cite{censusdata}. Travel rates $\tau$ for the infection flow matrix $A$ are derived from the Human Mobility Flow dataset \cite{kang2020multiscale} using the methodology in Section~\ref{Travel_Matrix}. We set $t_i = \frac{1}{3}$ for all $i$, assuming people spend one-third of their time outside.

\noindent
\textbf{Epidemic and Economic Parameters.} 
GDP data from the Bureau of Economic Analysis \cite{beagdpdata} is used to assess economic impact on counties. Economic loss ($z_i ^\mathrm{a}$ and $z_i ^\mathrm{s}$) is calculated using $g_i$, a county's GDP, as $z_i ^\mathrm{a} = z_i ^\mathrm{s} = g_i/g_{\max}$, where $g_{\max}$ is the highest GDP among all counties.

COVID-19 spread is modeled using parameters from \cite{birge2022controlling, 2020:Giordano_Colaneri_855_860}. The recovery rate $\gamma$ is set uniformly for all disease states ($r^a = r^s = r^q = \gamma = 0.2$). The rate of developing symptoms $\epsilon$ is 0.32. Transmission rates $\beta^a$ and $\beta^s$ follow \cite{2021:Ma_Olshevsky_Arxiv} with $\beta^a = \eta \beta^s$ and $\eta = 0.6754$ \cite{2020:Giordano_Colaneri_855_860}. $\beta^s$ is adjusted to match the observed initial growth rate, and $\beta^a$ is calculated based on $\eta$. Adjustments to $\beta$ ensure the infection flow matrix $A$ meets the required growth rate.

\noindent
\textbf{Initial Rates.} 
Initial susceptible rates $s(t_0)$ are set using county-level data on cumulative cases and deaths \cite{mitch_smith}, adjusted for unreported cases by dividing cumulative cases by 0.14 \cite{birge2022controlling, hortaccsu2021estimating}. The initial susceptible rate of node $i$ is $1 - \frac{c_i(t_0)}{0.14 N_i}$, where $c_i(t_0)$ is the cumulative cases for node $i$. Initial state $t_0$ is set to April 1, 2020.

Initial recovered cases are calculated by multiplying the total recovered ratio of US cumulative cases on April 1, 2020 $\frac{8878}{215215}$ \cite{2021:Ma_Olshevsky_Arxiv} by each node's cumulative infected population and adding deaths. Active cases, both asymptomatic and symptomatic, are calculated as the remainder of cumulative cases. We assume 14\% are symptomatic and 86\% are asymptomatic, with all rates adjusted by 0.14 to account for reporting inaccuracies.

\subsection{Massachusetts County-Level $14$ Node Network}\label{Mass_network_plots}

\noindent
\textbf{Optimizing Travel Rates.} 
We now validate the correctness of our optimal solution for problem \eqref{problem_formulation_P1} over the $14$ node network that represents the state of Massachusetts. 
As we mentioned in Section~\ref{mod_par_plots}, for constructing the travel rate matrix $\tau$ and the infection flow matrix $A$, we utilized the Human Mobility Flow dataset \cite{kang2020multiscale}. 

In Fig.~\ref{fig:f(tau)_MA_P1} we observe the value of $f(\tau) = \lambda_{\max} \left( M(t_0, \tau) \right)$ (cf. \eqref{problem_formulation_P1}) using the aforementioned PGD algorithm. For every budget $b$, we reach the local optimal point $\tau^*$ after a sufficient number of iterations. As $b$ increases, the optimal travel rates move away from the uncontrolled values while minimizing $f(\tau)$. So, as depicted in Fig.~\ref{fig:f(tau)_MA_P1}, we obtain lower $\lambda_{\max} (M(t_0, \tau^*))$ for larger values of $b$. 
Specifically, for $b > 20$, we have that $f(\tau^*)$ becomes less than $0$ making the epidemic network converge faster to a disease-free equilibrium. Note that the degree of restrictions $b$ depends primarily on the travel rates within the network. Therefore, the customization for a proper choice in each specific network is essential. 
\begin{figure}[t]
    \centering
    \includegraphics[width=0.35\textwidth]{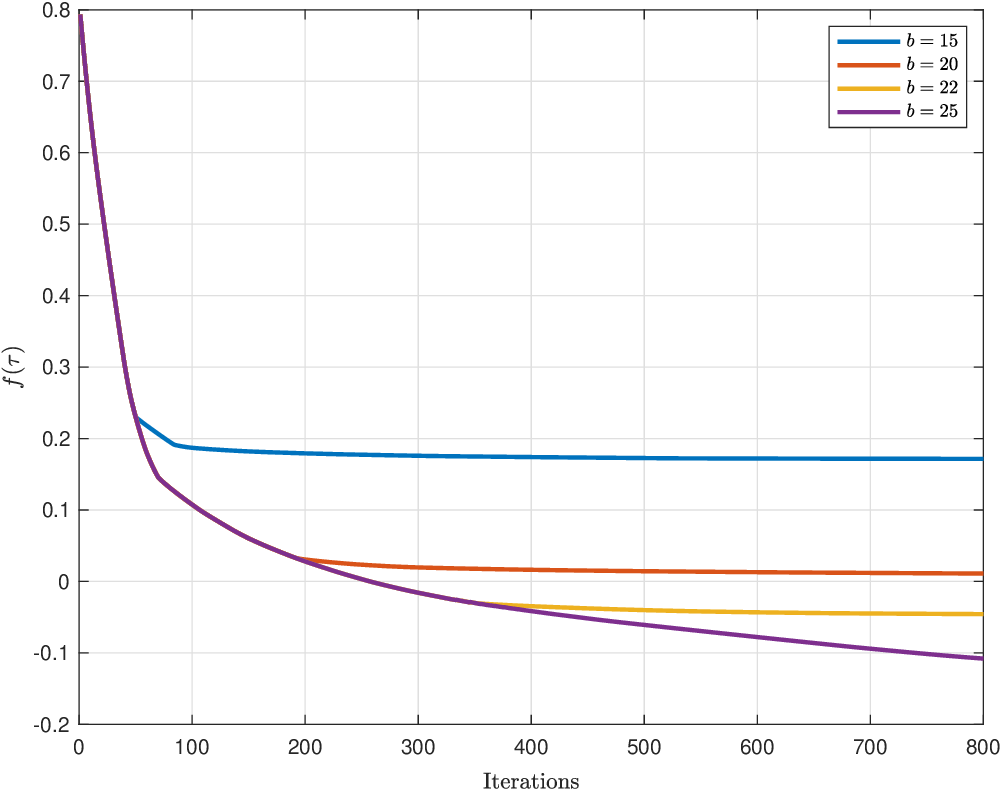}
    \caption{Optimal values of $f(\tau)$ in problem \eqref{problem_formulation_P1} for different budget parameters $b$ after changing the travel rates $\tau$ via \eqref{PGD_iteration}.} 
    \label{fig:f(tau)_MA_P1}
\end{figure}

In Fig.~\ref{fig:Cumulative_MA_P1} we observe the cumulative number of infected and recovered cases for the optimal travel rates obtained from PGD algorithm in Fig.~\ref{fig:f(tau)_MA_P1} within Massachusetts counties.
As expected, the optimal travel rates associated with larger $b$ leads to lower number of cumulative cases and a faster epidemic control due to the stricter lockdown measures imposed.

Similarly, Fig.~\ref{fig:Active_MA_P1} depicts the number of active infected cases for the optimal travel rates obtained from the PGD algorithm in Fig.~\ref{fig:f(tau)_MA_P1}.
Notably, we achieve a reduction rate of $-\alpha = - 0.0231$, which corresponds to halving infected cases every $30$ days, for the optimal travel rates obtained for $b > 20$.

\begin{figure}[t]
    \centering
    \includegraphics[width=0.35\textwidth]{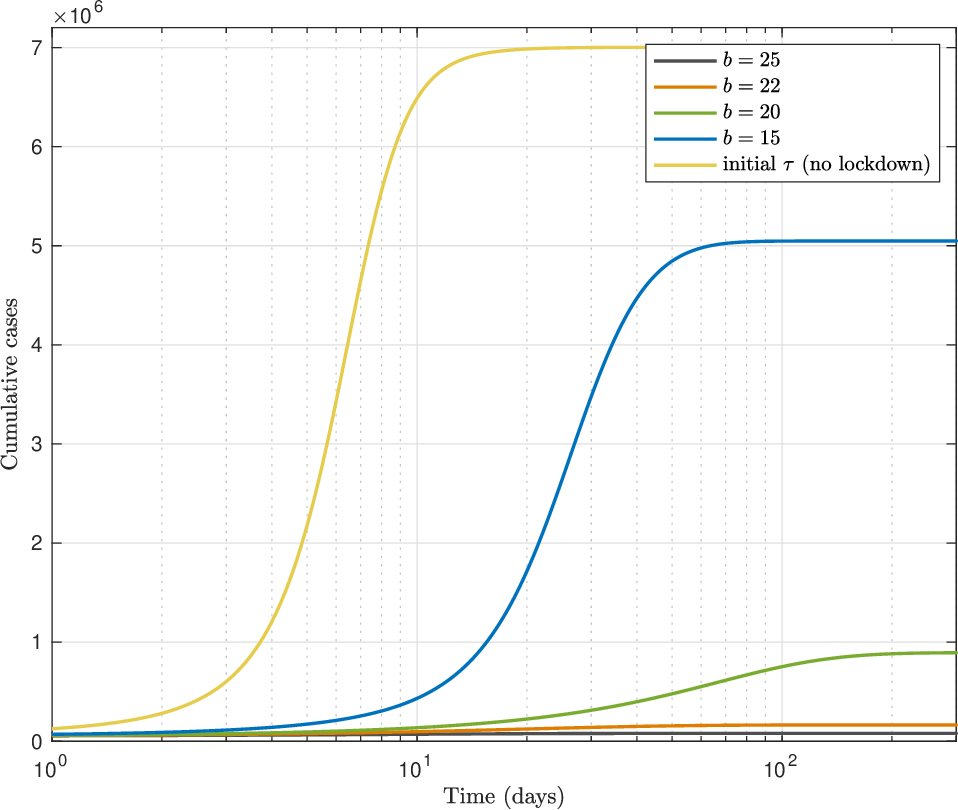}
    \caption{Number of cumulative cases (infected, quarantined, and recovered) for different constraints over the travel rates for the state of Massachusetts.} 
    \label{fig:Cumulative_MA_P1}
\end{figure}

\begin{figure}[t]
    \centering
    \includegraphics[width=0.35\textwidth]{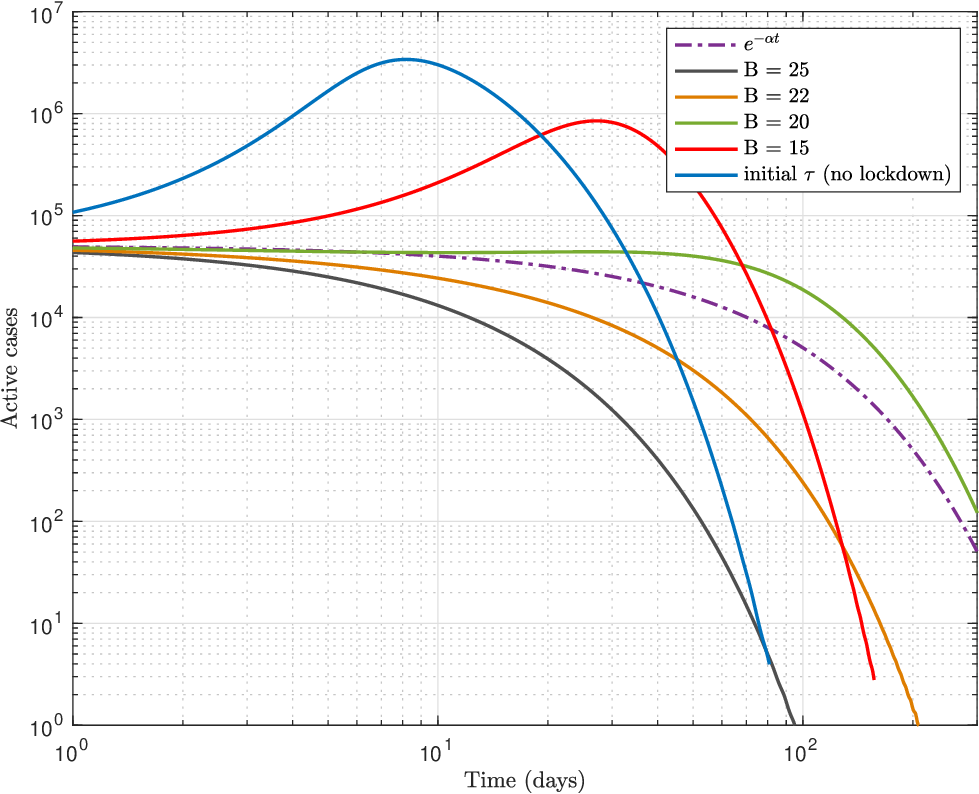}
    \caption{Number of active (asymptomatic and symptomatic infected) cases for different constraints over the travel rates for the state of Massachusetts.} 
    \label{fig:Active_MA_P1}
\end{figure}

\noindent
\textbf{Optimizing Quarantine Rates.} 
Fig.~\ref{fig:cumulative_MA} and Fig.~\ref{fig:active_without_q} compare the optimal quarantining rates calculated by our method with several other quarantining policies. In uniform quarantining, the quarantine rates of all locations are chosen to be the same and such that the total economic cost equals the cost associated with the optimal quarantine policy. In random quarantining however, the rates are randomly chosen from a uniform distribution and the parameters are chosen such that we have the same total economic cost as the optimal quarantine policy. In a uniformly bounded decline policy, quarantine rates are determined to ensure that the infection decay rates of all nodes remain within a specified bound, resulting in the same economic cost as our optimal policy. For the same economic costs, we observe that optimal quarantine rates minimize the number of infected (cumulative and active) cases in the network.

\section{Conclusions}\label{sec:conclusions}

In this paper, we presented a framework for epidemic control with two approaches. 
The first focuses on strategically reducing travel rates to contain the virus effectively, utilizing the maximum eigenvalue function and the PGD algorithm for optimization. 
The second approach enhances the SIR model with a quarantine strategy to minimize costs and decrease new infections rapidly through node-specific rates. 
We propose a solution that simplifies optimal quarantining into a weight-balancing problem and establish a link between optimization constraints and the epidemic's basic reproduction number. 
Finally, applying (Aug-PDGD) to our optimal quarantine problem ensures exponential stability of the solution.

\begin{figure}[t]
    \centering
    \includegraphics [width=0.35\textwidth]{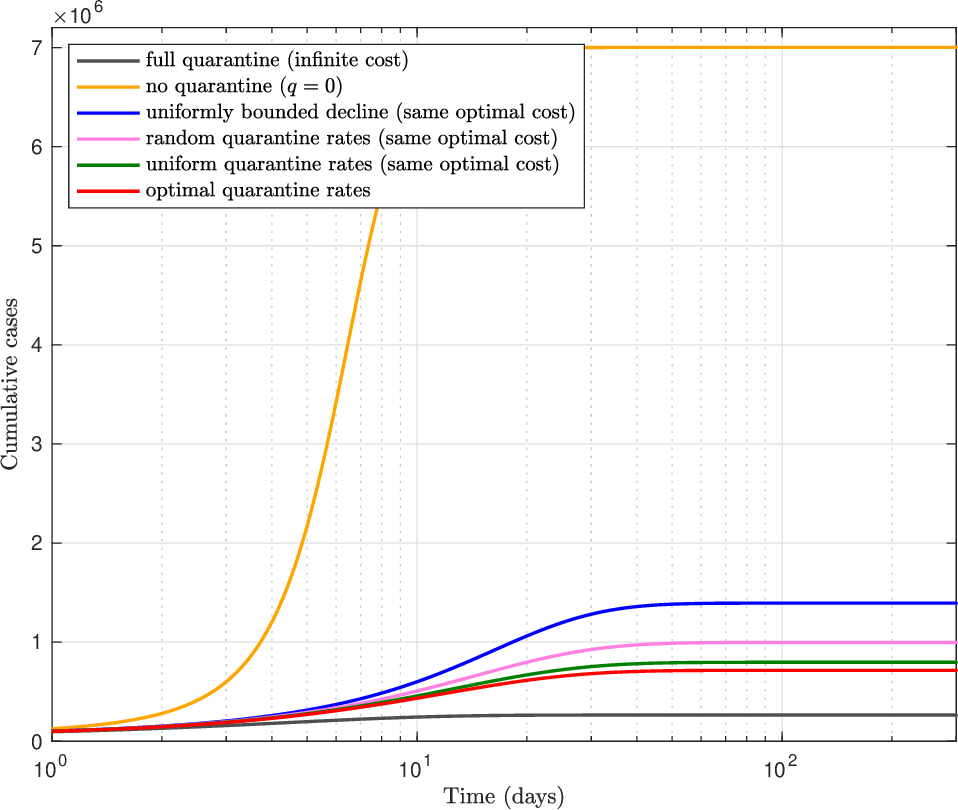}
    \caption{Number of cumulative cases (infected, quarantined, and recovered) assuming for different quarantine rates for the state of Massachusetts. For the optimal policy, $\alpha$ is set to $0.023$, which corresponds to halving the number of infected cases every $30$ days.} 
    \label{fig:cumulative_MA}
\end{figure}
\begin{figure}[t]
    \centering
    \includegraphics[width=0.35\textwidth]{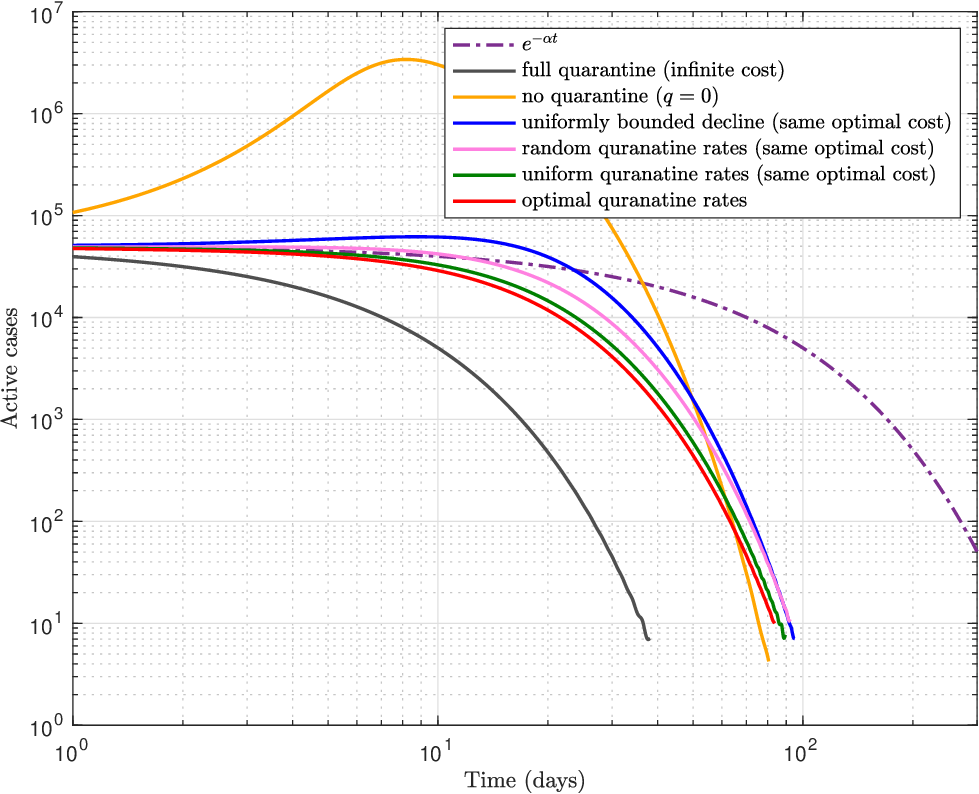}
    \caption{Number of active (asymptomatic and symptomatic infected) cases  for different quarantine policies for the state of Massachusetts. For the optimal policy, $\alpha$ is set to $0.023$, which corresponds to halving the number of infected cases every $30$ days.}
    \label{fig:active_without_q}
\end{figure}

Our numerical results highlight the efficiency of both approaches in controlling the epidemic. Although we constructed the network within the county-level framework, it's noteworthy that, given  data availability, these methods can be applied to much smaller nodes (e.g., zip code or census track levels).




\appendices

\section{Proof of Theorem~\ref{SIQR_unconstrained}}
\label{appendix:A}

Consider the problem formulation in \eqref{problem_formulation2}.
Setting $w = \bone - q$, we equivalently have 
\begin{equation} 
\begin{split} 
    &{\min}_{w} \quad \sum_{i=1}^{n} \frac{z_i^\mathrm{a}}{w_i^{\mathrm{a}}} + \frac{z_i ^s}{w_i^{\mathrm{s}}} \\
    &\text{s.t.} \quad \lambda_{\max} (A_0) \leq 0 \ , \\
    & \qquad \bzero \leq w \leq \bone \ ,
\end{split} 
\label{revised_formulation1}
\end{equation}
where from \eqref{matrix_for_convergence}
\begin{equation*}
A_0 =
\begin{pmatrix}
    \diag{(w^{\mathrm{a}})} + E'  & \beta^{\mathrm{s}} \diag{(s(t_0))} A\\
    \epsilon I & \diag{(w^{\mathrm{s}})} - (r^{\mathrm{s}} + 1 - \alpha) I
\end{pmatrix}, 
\end{equation*}
where $E' = \beta^{\mathrm{a}} \diag{(s(t_0))} A - (\epsilon + r^{\mathrm{a}} + 1 - \alpha) I$. 

We proceed by finding the range of the decay rate $\alpha$ that ensures the existence of a non-negative $w$ which satisfies the maximum eigenvalue constraint in~\eqref{revised_formulation1}. Let us first assume 
\begin{equation}
    \alpha < \min(r ^ s + 1, \ \epsilon + r^a + 1 - \max_i{(\beta^{\mathrm{a}}\diag{(s(t_0))} A)_{ii}}) ,
    \label{alpha_assumption1}
\end{equation}
where $\max_i{(\beta^{\mathrm{a}}\diag{(s(t_0))} A)_{ii}}$ is the maximum diagonal element of $\beta^{\mathrm{a}}\diag{(s(t_0))} A$, which is known and fixed. 
If \eqref{alpha_assumption1} does not hold, then matrix $A_0$ cannot be continuous time stable (i.e., the first constraint in \eqref{revised_formulation1} is not satisfied for any choice of $w$). 
More specifically, if \eqref{alpha_assumption1} is not satisfied for the decay rate $\alpha$, there exists no $d > \bzero$ satisfying $A_0 d \leq \bzero$, regardless of the choice of $w$ (see the first bullet of Lemma~\eqref{splitting_lemma}). 

Let us write
\begin{align} 
    A_0 &= \diag{(w)} +
    \begin{pmatrix}
        E' & \beta^{\mathrm{s}} \diag{(s(t_0))}A \\
        \epsilon I & -(r^{\mathrm{s}} + 1 - \alpha) I
    \end{pmatrix} \notag \\
    & = \diag{(w)} + B_0,
    \label{eq.18}
\end{align}
where $w = (w^{\mathrm{a}}, w^{\mathrm{s}})$ $\in \mathbb{R}^{2n}$, and $E' = \beta^{\mathrm{a}} \diag{(s(t_0))} A - (\epsilon + r^{\mathrm{a}} + 1 - \alpha) I$. 
From \eqref{eq.18}, we set $A_0$ to be the sum of a non-negative diagonal matrix $\diag{(w)}$ with $B_0$. 
By the Perron-Frobenius theorem, since adding $\diag{(w)}$ to $B_0$ cannot decrease the maximum eigenvalue of $B_0$, in order to have a feasible problem in \eqref{revised_formulation1}, it is necessary for $B_0$ to be Hurwitz. 
However, although being necessary, \eqref{alpha_assumption1} is not a sufficient condition for $B_0$ to be Hurwitz. 
In order to make $B_0$ Hurwitz, we implement the following approach. We can write
\begin{align*}
    B_0 = & \begin{pmatrix}
        \beta^{\mathrm{a}} \diag{(s(t_0))} A - (\epsilon + r^{\mathrm{a}} + 1)I  & \beta^{\mathrm{s}} \diag{(s(t_0))}A \\
        \epsilon I & -(r^{\mathrm{s}} + 1) I 
    \end{pmatrix}\nonumber \\ 
    & + \alpha I = C_0 + \alpha I.
\end{align*}
Again, to have $B_0$ Hurwitz, the matrix $C_0$ needs to be Hurwitz and $\alpha$ must be chosen such that adding it to the diagonal elements of $C_0$ does not make $C_0$ unstable. So, to have a feasible problem, $\alpha$ should satisfy
$\alpha < - \lambda_{\max} (C_0) $.

Having $B_0$ Hurwitz, implies that $B_0$ is invertible and $-B_0^{-1} \geq \bzero$ (\cite[ Thm.~$10.3$]{FB-LNS}). 
Using Lemma~\ref{splitting_lemma}, having $A_0$ continuous-time stable is equivalent to 
$
    -B_0^{-1} \diag{(w)}
$
being discrete-time stable. Since changing the order of the product of two matrices doesn't affect the eigenvalues of the product, we equivalently want the matrix
$\diag{(w)}(-B_0^{-1})$ to be discrete-time stable. 
Applying part~$(3)$ of Lemma~\ref{splitting_lemma}, this is equivalent to
$-B_0^{-1} - {\diag{(w)}}^{-1}$
being continuous-time stable. 

For a non-negative $w$, let us set $v_i = {w_i}^{-1}$, then\eqref{revised_formulation1} without considering $w \leq \bone $ is equivalent to 
\begin{equation} 
    \begin{split} 
        &{\min}_{v} \quad z^T \, v \\
        &\text{s. t.} \quad \lambda_{max}(-B_0^{-1} - {\diag{(v)}})\leq 0, \\ 
       &  v \geq \bzero
    \end{split}
    \label{reduced_revised_formulation1}
\end{equation}
where $-B_0^{-1}$ is Metzler. Due to the strong connectivity of A, the two upper blocks in $B_0$ are strongly connected. The two lower diagonal blocks ensure the strong connectivity of $B_0$. Hence, $-B_0$ is irreducible as well. The inverse of an irreducible matrix is also irreducible and we have $-B_0^{-1}$ being irreducible.
The objective of \eqref{reduced_revised_formulation1} is to find the optimal solution $v^*$ such that $z^T v^*$ is minimized and the matrix $-B_0^{-1} - {\diag{(v^*)}}$ is Hurwitz. 
Applying, \cite[Theorem~$3$]{2022:Ron_Bullo_2731_2736}, we have that \eqref{reduced_revised_formulation1} can be reduced to a matrix balancing problem. 
More specifically, consider $d^* \in \mathcal{R}^{2n}_{>0}$ such that the matrix 
$$
(\diag{(d^*)})^{-1} \diag{(z)} (-B_0^{-1}) \diag{(d^*)}
$$ 
is weight balanced.  
From \cite[Theorem~$3$]{2022:Ron_Bullo_2731_2736} we have that  \eqref{reduced_revised_formulation1} has a unique optimal solution 
\begin{equation}\label{optimal_sol_weight_bal}
    v^* = (\diag{(d^*)})^{-1} (-B_0^{-1}) \diag{(d^*)} \mathbf{1} \ .  
\end{equation} 
So, from the above relation we have shown that $\eqref{revised_formulation1}$ is equivalent to a matrix balancing problem for $w \geq \bzero$, or equivalently $q\leq \bone$.

To complete our analysis, we also require $w \leq \bone$ (or equivalently $q \geq \bzero$). 
Let us see the solution $v^*$ in more details. As discussed previously,  $B_0$ is a Hurwitz matrix with non-negative off-diagonal elements and negative diagonal elements. Suppose $m$ is the absolute value of the most negative diagonal element of $B_0$, or 
\begin{equation*}\label{m}
m = \max_i |(B_0)_{ii}| > 0 \ .
\end{equation*}
Let us write
    $B_0 = B - mI$, 
where $B$ is a non-negative matrix. From the Perron Frobenius theorem we have $\rho (B) >0$ and $\rho(B) - m$ is an eigenvalue of the Hurwitz matrix $B_0$. 
\begin{equation} \label{B_0_inverse}
\begin{aligned}
    B_0^{-1} &= (B-m I)^{-1} = \frac{1}{m}(\frac{1}{m}B - I)^{-1} \\
    &= - \frac{1}{m}\sum_{n=0}^{\infty} \left(\frac{1}{m} B\right)^n = - \frac{1}{m}\sum_{n=0}^{\infty} \left(I + \frac{1}{m} B_0\right)^n.
\end{aligned}
\end{equation}
By substituting~\eqref{B_0_inverse} in~\eqref{optimal_sol_weight_bal}, $v^*$ can be derived as
\begin{equation}
    v^* = \frac{1}{m}(\diag{(d^*)})^{-1} \sum_{n=0}^{\infty} \left(I + \frac{1}{m} B_0\right)^n\diag{(d^*)} \mathbf{1} \ .  
\end{equation} 
Thus, to have $w \leq \bone$, we require $v^* \geq \bone$ or equivalently, the sum of row elements of $$ (\diag{(d^*)})^{-1} \sum_{n=0}^{\infty} \left(I + \frac{1}{m} B_0\right)^n\diag{(d^*)}
$$ should be greater or equal to $m$. Indeed, defining $\sum_{n=0}^{\infty} \left(I + \frac{1}{m} B_0\right)^n = Q$, we need 
\begin{equation}\label{solution_inequality}
    \sum_{j=1}^{2n} \frac{d^*_j}{d^*_i}Q_{ij} \geq m \qquad \forall i \in \{1,2,\dots, 2n\} \ .
\end{equation}

Notice that $\frac{1}{m}B_0$ is a matrix with non-negative off-diagonal elements and negative diagonal elements of $-1 \leq \frac{(B_0)_{ii}}{m} < 0$. Thus, $(I + \frac{1}{m} B_0)$ is a non-negative matrix with diagonal elements, and 
$$0 \leq (I +\frac{1}{m}B_0)_{ii} < 1 \ .$$
Hence, 
\begin{equation}
    Q \geq I + [I + \frac{1}{m}B_0] + [(I + \frac{1}{m}B_0)^2] = \Bar{Q}
    \label{eq.34}
\end{equation}
Letting $(I + \frac{1}{m} B_0)_{ii} = 0$, we have 
\begin{align}
        (I + \frac{1}{m} B_0) \geq
        \begin{pmatrix}
            \frac{1}{m}\Tilde{A} & \frac{\beta^{\mathrm{s}}}{m} \diag{(s(t_0))} A \\
            \frac{\epsilon}{m} I & 
            \bzero
        \end{pmatrix} \ ,
        \label{eq.32}
\end{align}
where 
$
\Tilde{A} = \beta^{\mathrm{a}} \diag{(s(t_0))} A
$ that its diagonal elements are replaced with zero. Substituting \eqref{eq.32} into \eqref{eq.34},
we derive a lower bound of $\Bar{Q}$ as $\Tilde{Q}$:
\begin{align}
    \Tilde{Q}= & \begin{pmatrix}
        \Tilde{Q}_{11} & \Tilde{Q}_{12} \\
        \frac{\epsilon}{m} I  +  \frac{\epsilon}{m^2} \Tilde{A} & 
        I + \frac{ \epsilon \beta^{\mathrm{s}}}{m^2} \diag{(s(t_0))} A
    \end{pmatrix},
\end{align}
where $$
\Tilde{Q}_{11} =     
        I + \frac{\Tilde{A}}{m}  + \frac{1}{m^2}(\Tilde{A}^2 + \epsilon \beta^{\mathrm{s}} \diag{(s(t_0))} A),
$$
$$
\Tilde{Q}_{12}=
        \frac{\beta^{\mathrm{s}}}{m} \diag{(s(t_0))} A  + \frac{\beta^{\mathrm{s}} }{m^2}\Tilde{A} \diag{(s(t_0))} A \ .
$$ 

From the definition of $B_0$ in~\eqref{eq.18}, $m$ can be written as
\begin{equation*}
\begin{aligned}
    m = \max\Big(&r^s + 1 - \alpha, \\
    & \epsilon + r^a + 1 - \alpha - \min\left(\beta^{\mathrm{a}} \operatorname{diag}(s(t_0)) A \right)_{ii}\Big) \ .
\end{aligned}
\end{equation*}

Substituting Assumption~\ref{assump1} into the definition of $\Tilde{Q}$, we have 
\begin{equation*}
    \begin{aligned}
       \Tilde{Q}_{ii} &= 1 + \frac{1}{m^2}((\Tilde{A}^2)_{jj} + \epsilon \beta^{\mathrm{s}} s_i(t_0) A_{jj}) \\
       &\geq 1 + \frac{\epsilon}{m^2} \beta^{\mathrm{s}} s_i(t_0) A_{jj} \geq 1 + \frac{x}{m^2}\\
       &\geq m  \ , \quad \, \quad i \in \{1,\dots,n\} \ ,
    \end{aligned}
\end{equation*}  
and
\begin{equation*}
    \begin{aligned}
       \Tilde{Q}_{ii} &= 1 + \frac{1}{m^2} \epsilon \beta^{\mathrm{s}} s_j(t_0) A_{jj}\geq 1 + \frac{x}{m^2} \geq m \\
       &\geq m \ , \quad \, \quad i \in \{n+1,\dots,2n\} \ .  
    \end{aligned}
\end{equation*} 
Finally,
\begin{equation}
    \sum_{j=1}^{2n} \frac{d^*_j}{d^*_i}Q_{ij} \geq Q_{ii}
    \geq \Bar{Q}_{ii} \geq
    \Tilde{Q}_{ii} \geq m \ ,
\end{equation}
for all $i \in \{1,2,\dots, 2n\}$. Hence, $v^* \geq \bone$, and equivalently $q^* \geq \bzero$. Subsequently, substituting the definitions of $v$ and $w$,
the optimal quarantined rates are calculated as
\begin{equation*}
    q^* = \mathbf{1} + 1./\big({(\diag{(d^*)})}^{-1} B_0^{-1} (\diag{(d^*)}) \mathbf{1}\big) \ . 
\end{equation*}

\section{Proof of Theorem~\ref{reprod_theor_proof}}
\label{appendix:B}

We will establish that initiating from $R_0 < 1$ leads us to $\lambda_{\max} (M(t_0)) \leq 0$. Assume 
\begin{equation*}
   R_0 \leq 1 \ . 
   \label{condition1}
\end{equation*}
We have
\begin{equation*}
M(t_0) = 
    \begin{pmatrix}
    E & \beta^{\mathrm{s}} \diag{(s(t_0))} A  \\
    \epsilon I & - (r^{\mathrm{s}} + \diag{(q^{\mathrm{s}})}) I \\
    \end{pmatrix}, 
\end{equation*}
where $E = \beta^{\mathrm{a}} \diag{(s(t_0))} A - (\epsilon + r^{\mathrm{a}} + \diag{(q^{\mathrm{a}}})) I$. 
Let us write $M(t_0)$ as 
\begin{equation*}
    M(t_0) = F + V,
\end{equation*}
where 
\[
F = \begin{pmatrix}
    \beta^{\mathrm{a}} \diag{(s(t_0))} A  & \beta^{\mathrm{s}} \diag{(s(t_0))} A  \\
    \bzero  & \bzero \\
    \end{pmatrix},
\]
and 
\[
V = \begin{pmatrix}
    -(\epsilon + r^{\mathrm{a}})I - \diag{(q^{\mathrm{a}})}) & \bzero  \\
    \epsilon I & - (r^{\mathrm{s}} + \diag{(q^{\mathrm{s}})}) I \\
    \end{pmatrix}.
\]
Note that $F$ is nonnegative and $V$ is Metzler and Hurwitz. Using the definition of $R_0$ in \cite[Definition~$1$]{smith2022convex} we have
\begin{equation}\label{eq.35}
    R_0 = \rho(F V ^{-1}) =
    \rho( -F V ^{-1}) \leq 1.
\end{equation}
Utilizing the early results of Thm.~\ref{SIQR_unconstrained}, \eqref{eq.35} is equivalent with
\begin{equation*}
    \lambda_{\max} (F + V) \leq 0,
\end{equation*}
(i.e., $F + V$ is Hurwitz).
Hence, we have
\begin{equation*}
    \lambda_{\max} (M(t_0)) = \lambda_{\max} (F + V) \leq 0.
\end{equation*}
Therefore, the set $\lambda_{\max} (M(t_0)) \leq 0$ is equivalent with the set $R_0 \leq 1$.

\section{} \label{appendix:C}
Given the left and right eigenvectors of \( \lambda_{\max} (M(t_0,q)) \) as \( v_{\max}^T \) and \( u_{\max} \), respectively, we argue that 
\[ \left\| \frac{v_{\max} \circ u_{\max}}{v_{\max}^T u_{\max}} \right\| \]
is bounded for any choice of \( M(t_0, q) \). Without loss of generality, we may assume the left and right eigenvectors are normalized to have unit norm. In that case, the norm of the numerator is automatically bounded.

Applying Lemma~\ref{perron_frob_3} to \( M(t_0,q) \) and \( M(t_0,q)^T \), we have that \( u_{\max} \) and \( v_{\max} \) are positive as are the respective dominant eigenvectors. Hence, the dot product of the corresponding Perron-Frobenius left and right eigenvectors is always positive. This also implies that the denominator is bounded away from zero, for otherwise by compactness we could find a matrix $M$ such that $v_{\max}^T u_{\max} = 0$, contradicting positivity of these vectors. Therefore, 
$\|\nabla g_1(q)\|$ is bounded.


\section{Proof of Theorem~\ref{f_lipschitz}} \label{appendix:D}
It is obvious that the gradient of $f(q)$ is Lipschitz as long as every $q_i(0)$ is bounded away from $1$. We thus need to argue that the primal-dual dynamics keep all $q_i(0)$ uniformly bounded away from one. 

We start by expressing the primal dynamics as given by~\eqref{q_dot}:
\begin{equation}\label{q_dot2}
    \dot{q}(t) = -\nabla f(q(t)) - \sum_{i=1}^{4n+1} [\rho g_i(q(t)) + \lambda_i(t)]_{+} \nabla g_i(q(t)).
\end{equation}
Here, \(\lambda(t) \in \mathbb{R}_{\geq 0}^{(4n+1)}\) is a vector that remains in a compact region of space for all \(t \geq 0\) by Proposition~\ref{trajectory_set}. Therefore, using~\eqref{lipschitz_constraints}, we can state that the second part of~\eqref{q_dot2} is bounded.

Note that as \( q_i \rightarrow 1 \), we have  \(|\nabla_{i} f(q(0))|\) approaches \(\infty\), and \(\dot{q}_i < 0\). This immediately implies there is some \( q' < 1 \) such that  the region \([0, q']^n\) remains invariant under the differential equation~\eqref{q_dot2}. Remaining within this region implies the Lipschitz continuity of \(\nabla f(q)\).

\section{Proof of Proposition~\ref{regularity_prop}} \label{appendix:E}

Suppose \( q^* \) is a local minimum of~\eqref{problem_formulation2}. Therefore, \( q^*_i \neq 1 \) for all \( i \in \{1, 2, \dots, 2n\} \). This implies the active set is
\[
\mathcal{I} = \{ j : g_j(q^*) = 0, \ j \in J \} \ ,
\]
where \( J = \{1, 2, \dots, 2n+1\} \).

First, assume \(\mathcal{I} \neq J\). From the definition provided for \( g_1 \), \(\nabla g_1 = \frac{v_{\max} \circ u_{\max}}{v_{\max}^T u_{\max}} \) is an element-wise positive vector by the same analysis provided in Appendix~\ref{appendix:D}. It is therefore evident that any \( 2n \) vectors of \( g_j \), for all \( j \in J \), would have linearly independent gradients at any point \( q^* \).

Now, consider \(\mathcal{I} = J\). This implies \( q^* = \mathbf{0} \). However, for our model of the outbreak in \eqref{problem_formulation2}, setting all quarantine rates equal to zero implies having no interventions. Hence, the constraint function \( g_1 = \lambda_{\max}(M(t_0, q = \mathbf{0})) + \alpha \) would not be active. Hence, \(\mathcal{I}\) cannot include all \( j \in J \), and this proves the non-linearity of the active set at any local minimum \( q^* \).

\bibliographystyle{IEEEtran}
\bibliography{autosam}

\begin{IEEEbiography}[{\includegraphics[width=1in,height=1.25in,clip,keepaspectratio]{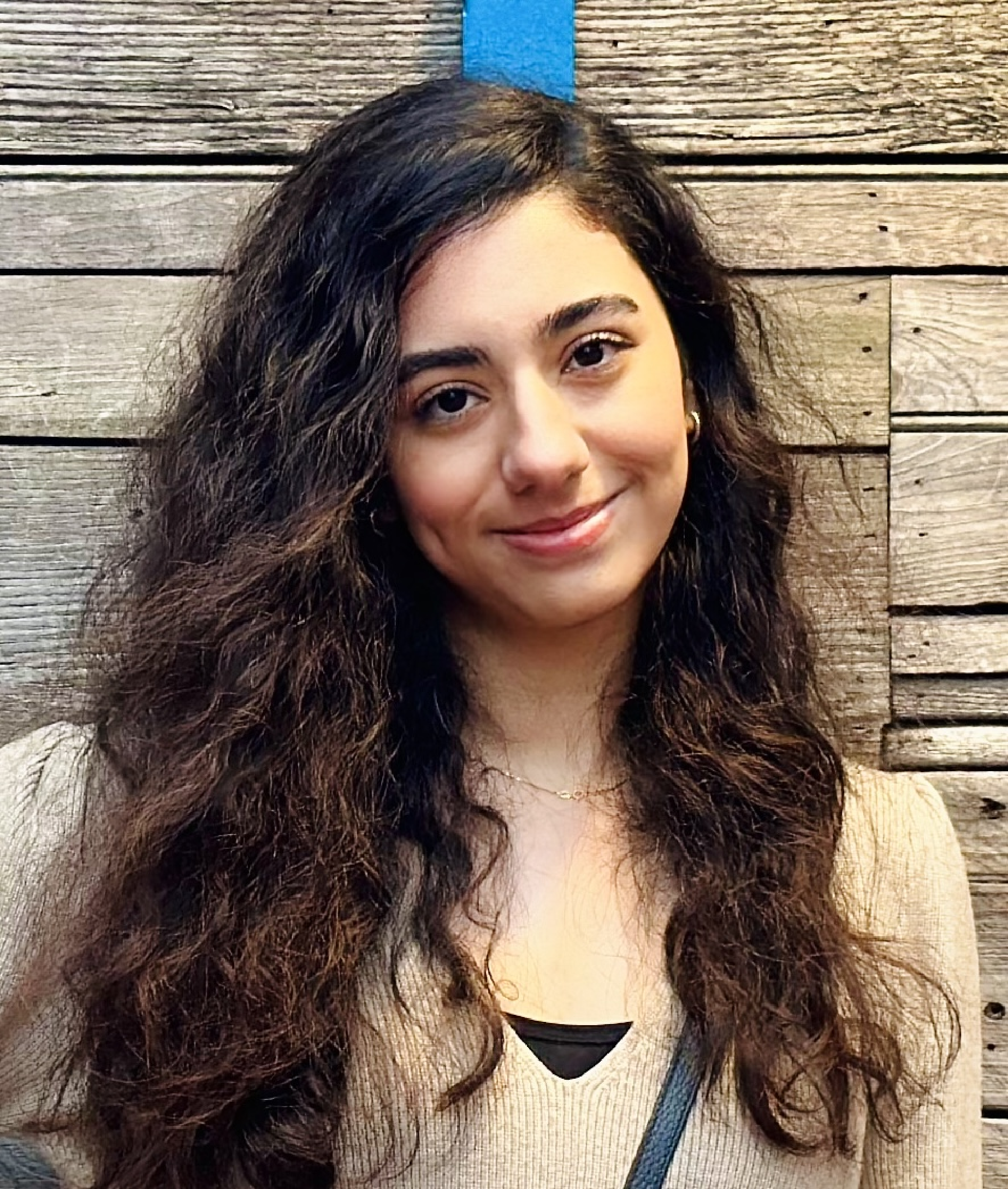}}]{Mahtab Talaei}
is currently pursuing her Ph.D. in the Division of Systems Engineering at Boston University in Boston, MA. She obtained her B.Sc. and M.Sc. degrees in Electrical Engineering from Isfahan University of Technology (IUT) in Isfahan, Iran, in 2019 and 2022, respectively. Her research primarily involves utilizing optimization and machine learning techniques to create predictive models for healthcare applications.
\end{IEEEbiography}

\begin{IEEEbiography}[{\includegraphics[width=1in,height=1.25in,clip,keepaspectratio]{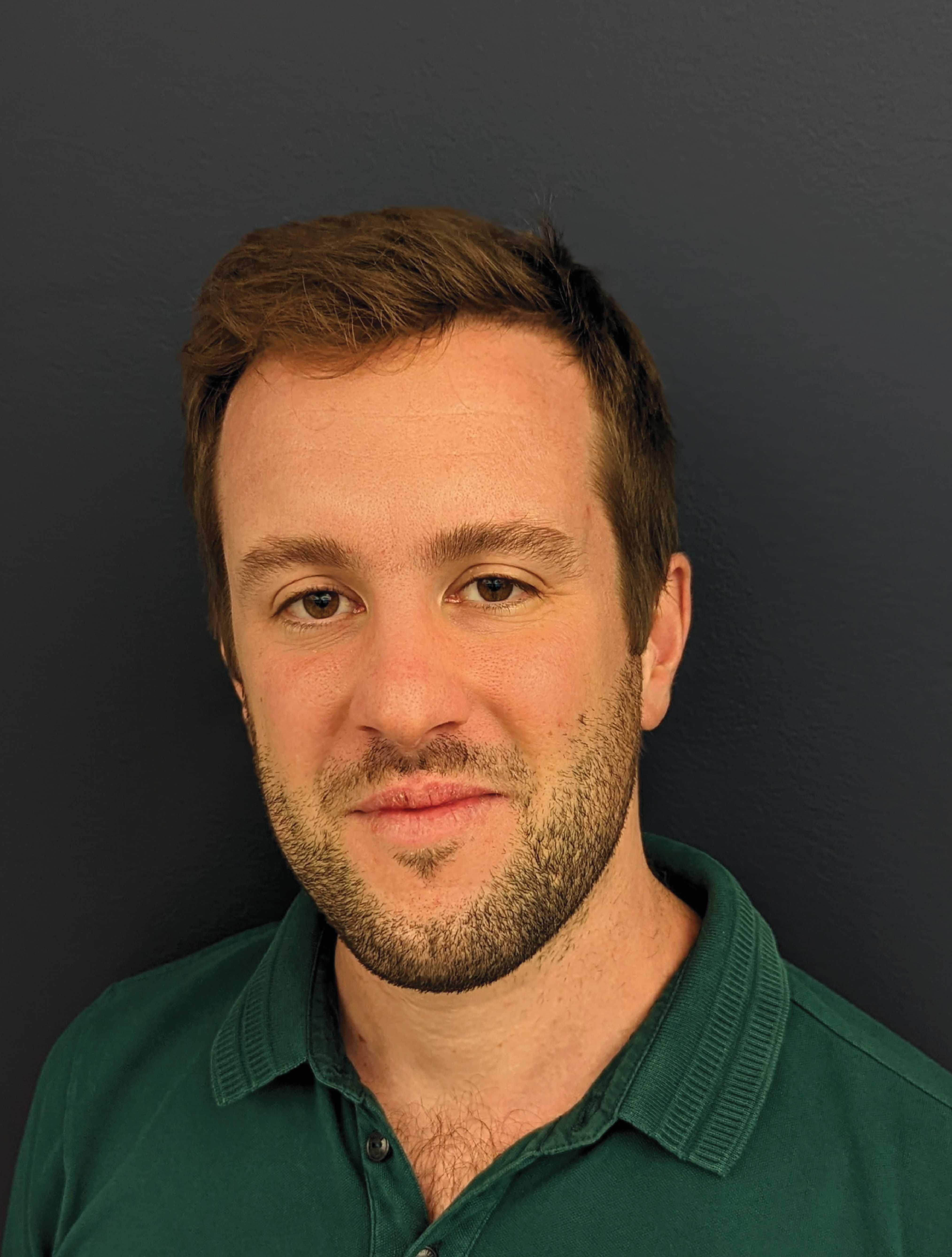}}]{Apostolos I. Rikos} (M'16) is an Assistant Professor at the Artificial Intelligence Thrust of the Information Hub, The Hong Kong University of Science and Technology (Guangzhou), Guangzhou, China. 
He is also affiliated with the Department of Computer Science and Engineering, The Hong Kong University of Science and Technology, Clear Water Bay, Hong Kong, China.
He received his B.Sc., M.Sc., and Ph.D. degrees in Electrical Engineering from the Department of Electrical and Computer Engineering, University of Cyprus in 2010, 2012, and 2018, respectively.
In 2018, he joined the KIOS Research and Innovation Center of Excellence in Cyprus, where he was a Research Lecturer. 
He joined the Division of Decision and Control Systems at KTH Royal Institute of Technology as a Postdoctoral Researcher in 2020 and the Department of Electrical and Computer Engineering, Division of Systems Engineering, at Boston University as a Postdoctoral Associate in 2023.
His research interests are in the areas of distributed optimization and learning, distributed network control and coordination, privacy and security, and algorithmic design.
\end{IEEEbiography}

\vskip -2\baselineskip plus -1fil
\begin{IEEEbiography}[{\includegraphics[width=1in,height=1.25in,clip,keepaspectratio]{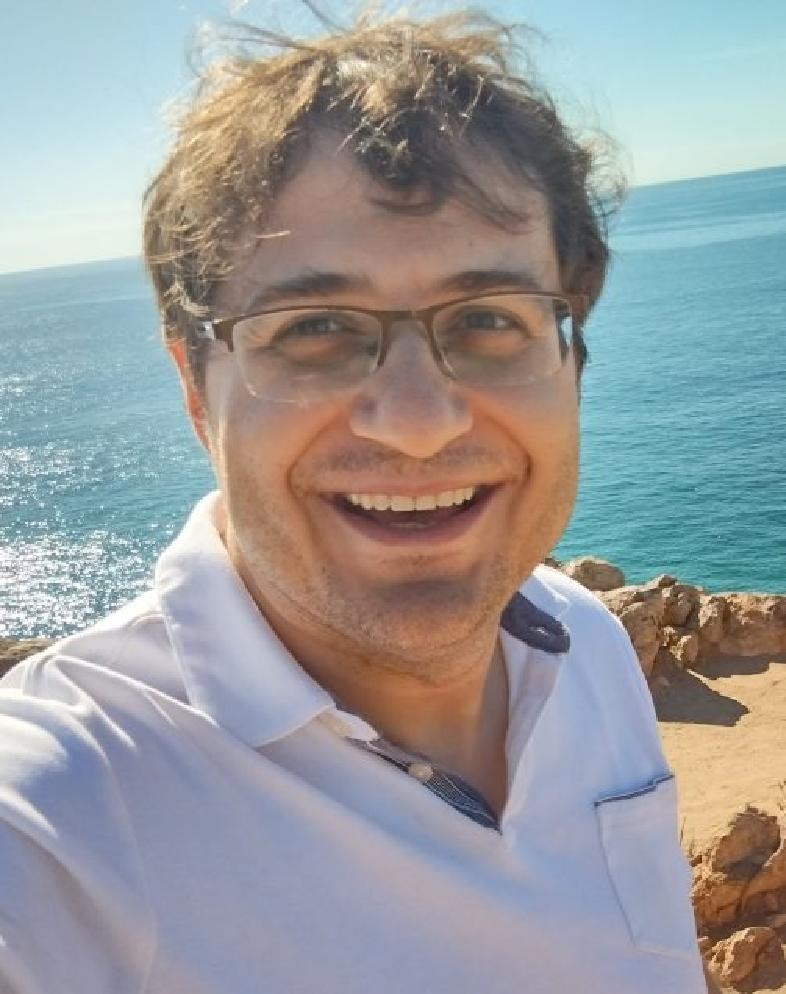}}]{Alex Olshevsky}
(Member, IEEE) received the B.S. degrees in Applied Mathematics and Electrical Engineering from the Georgia Institute of Technology, Atlanta, GA, USA, both in 2004, and the M.S. and Ph.D. degrees in electrical engineering and computer science from the Massachusetts Institute of Technology, Cambridge, MA, USA, in 2006 and 2010, respectively.
He is currently an Associate Professor with the Department of Electrical and Computer Engineering, Boston University, Boston, MA, USA.
His research interests include control systems, optimization, and net- work science.
He was the recipient of the National Science Foundation CAREER Award, the Air Force Young Investigator Award, the ICS Prize from INFORMS for best paper on the interface of operations research and computer science, and the SIAM Paper Prize for annual paper from the SIAM Journal on Control and Optimization chosen to be reprinted in SIAM Review.
\end{IEEEbiography}

\vskip -2\baselineskip plus -1fil
\begin{IEEEbiography}[{\includegraphics[width=1in,height=1.25in,clip,keepaspectratio]{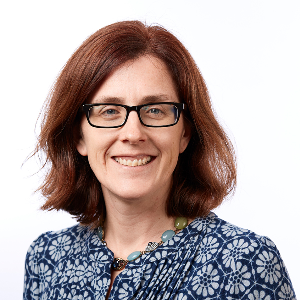}}]{Laura F. White} is a Professor of Biostatistics and Associate Director of the Population Health Data Science Program at the Boston University School of Public Health. She received her BS in Mathematics and Statistics from Utah State University and her PhD in Biostatistics from Harvard University. She co-directs the graduate program in Biostatistics at BU and is an co-director of the Data Science and Surveillance Core of the Center on Emerging Infectious Diseases. Her research interests are in developing novel statistical approaches to understanding infectious disease transmission dynamics and disease burden. Her work is supported by the NIH and CDC.  
\end{IEEEbiography}
\vskip -2\baselineskip plus -1fil
\begin{IEEEbiography}[{\includegraphics[width=1in,height=1.25in,clip,keepaspectratio]{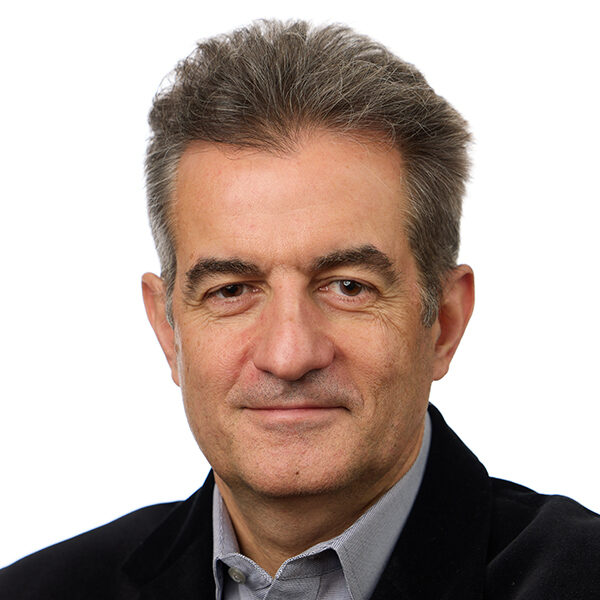}}]{Ioannis Ch. Paschalidis}
(Fellow, IEEE) received a diploma in Electrical and Computer Engineering (ECE) from the National Technical University of Athens, Greece, in 1991, and the MS and PhD degrees in Electrical Engineering and Computer Science (EECS), from the Massachusetts Institute of Technology (MIT), Cambridge, Massachusetts, in 1993 and 1996, respectively. In September 1996, he joined Boston University where he has been ever since. He is currently a Distinguished Professor of Engineering and the director of the Hariri Institute for Computing. 
His current research interests lie in the fields of optimization, control, stochastic systems, robust learning, and computational medicine/biology. His has been recognized with am NSF CAREER award, several best paper awards, and an IBM/IEEE Smarter Planet Challenge Award. He was an invited participant at the 2002 Frontiers of Engineering Symposium organized by the National Academy of Engineering, at the 2014 National Academies Keck Futures Initiative Conference, and at a 2024 National Academies Symposium on Alzheimer’s disease. He is a Fellow of IEEE, IFAC, and the Asia-Pacific AI Association. From 2013 to 2019 he was the founding Editor-in-Chief of the IEEE Transactions on Control of Network Systems and he is the General Co-Chair of the 2025 IEEE Conference on Decision and Control.
\end{IEEEbiography}

\end{document}